\PassOptionsToPackage{x11names}{xcolor}
\documentclass[12pt,a4paper,twoside]{article}

\usepackage{array}
\usepackage{amsmath}
\usepackage{tabularx}
\usepackage{txfonts}
\usepackage{enumerate}
\usepackage{cite}
\usepackage{float}
\usepackage{xspace}
\usepackage[x11names]{xcolor}
\usepackage{multirow}

\usepackage{lineno}

\usepackage{amssymb}

\usepackage[T1]{fontenc}
\usepackage[cp1250]{inputenc}

\usepackage[margin=2cm]{geometry}

\usepackage{changes}
\usepackage[colorinlistoftodos, textwidth=30mm, shadow]{todonotes}

\DeclareMathOperator{\e}{\mathrm{e}}
\DeclareMathOperator{\dd}{\mathrm{d\!}}

\renewcommand{\Re}{\textrm{Re}}

\newcommand{\R}{\varmathbb{R}}

\newcommand{\Cb}{\mathcal{C}}

\newtheorem{theorem}{Theorem}[section]

\newtheorem{prop}[theorem]{Proposition}

\newenvironment{proof}[1][]{\par\noindent\textbf{Proof #1: }}{\hfill\rule{1.3ex}{1.3ex}\\\par}

\numberwithin{equation}{section}

\begin{document}
\pagestyle{myheadings}
\markright{Piotrowska,~M.J.: An immune system--tumour interactions model with discrete time delay}

\noindent\begin{tabular}{|p{\textwidth}}
	\Large\bf An immune system --- tumour interactions model with discrete time delay: model analysis and validation\\\vspace{0.01cm}
    \it Piotrowska, Monika Joanna\\\vspace{0.02cm}
\it\small Institute of Applied Mathematics and Mechanics,\\
\it\small Faculty of Mathematics, Informatics and Mechanics,\\
\it\small University of Warsaw, Banacha 2, 02-097 Warsaw, Poland\\\vspace{0.01cm}
\footnotesize  \texttt{monika@mimuw.edu.pl},\\
    \multicolumn{1}{|r}{\large\color{orange} Research Article} \\
	\\
	\hline
\end{tabular}
\thispagestyle{empty}

\tableofcontents
\noindent\begin{tabular}{p{\textwidth}}
	\\
	\hline
\end{tabular}

\begin{abstract}

In this article a generalized mathematical model describing the interactions between malignant tumour and immune system with discrete time delay incorporated into the system is considered. Time delay represents the time required to generate an immune response due to the immune system activation by cancer cells. The basic mathematical properties of the considered model, including the global existence, uniqueness, non-negativity of the solutions, the stability of steady sates and the possibility of the existence of the stability switches, are investigated when time delay is treated as a bifurcation parameter. The model is validated with the sets of the experimental data and additional numerical simulations are performed to illustrate, extend, interpret and discuss the analytical results in the context of the tumour progression. 
\end{abstract}


\section*{Acknowledgements}
This work was supported by the Polish Ministry of Science and Higher Education, within the Iuventus Plus Grant: ''Mathematical modelling of neoplastic processes" grant No. IP2011 041971.


\section{Introduction}

Cancer immunology is a new trend in immunology that focuses on the variety of interactions between the 
immune system and cancer cells. The aim of cancer immunology is to discover new cancer immunotherapies to 
fully cure or at least to retard the progression of the disease. 
According to recent reports~\cite{Schreiber1999} lymphocytes sensitive to tumour antigens are often present in the patient's body, providing a potentially natural way to prevent tumour growth.
On the other hand, it is also observed that the full immune response usually does not occur
and so the natural protection afforded by the immune system is ineffective.

In general the interactions between tumour cells and immune system are very complex and involve a number of different kinds of molecules, immune cells and events, \cite{Alberts2008}. 
The immune response against tumour antigens involves different kinds of immune cells: B-Lymphocytes producing and secreting antibodies into the blood or placing them on their surfaces,  
cytotoxic T-Lymphocytes --  the effector cells  that destroy the 
antigens and  thus the antigen-bearing cells, T helper-lymphocytes secreting interleukins, hence stimulate both 
T and B cells to divide and finally T-suppressor cells that end the immune response.
All these cells are produced in the body from a small fraction of the precursor cells. Whenever antigens are 
presented to the precursor cells they are activated, start to proliferate and then differentiate into effector or 
memory cells. Unfortunately, these mechanism do not guarantee the fully efficient response of the immune 
system due mainly to the tumour's own defence system.
Tumour cells protect themselves from contact with lymphocytes by, amongst other methods, surface 
expression of ligands which initiate the apoptotic signal in the cytotoxic T-Lymphocytes, thus killing off the 
cytotoxic T-Lymphocytes before the necessary cell-contact with the tumour cells can be initiated.

Clearly, due to the complexity of the considered problem systems describing most of the tumour-immune system interactions would involve a large number of variables and thus equations, as in \cite{Prikrylova1992}, 
where a model of 12 equations  is proposed and numerically analysed or in \cite{Kuznetsov1994}, where a 
system of 5 equations describing the response of the effector cells to the growth of tumour cells is studied. 
More recently, systems of four or more equations investigating the immunotherapy and its improvement for 
high grade gliomas (\cite{nkykvvza08cancer,agur10siamappl}), predicting outcomes of prostate cancer 
personalized immunotherapy (\cite{Kronik2011}) or humoral mediated immune response~\cite{Feyissa2013} were proposed.
Also well known immune system Marchuk's model,~\cite{fomar00, fomar97, uf09jmaa}, describing the immune system dynamics which was applied to the description of the tumour growth phenomenon, see e.g.~\cite{uf02jtb}. 
However, to better understand the main underlying mechanism of immune response to the presence of the 
tumour cells complex models were simplified and thus reduced into two- or three-equation and less detailed ones, see e.g. \cite{Kirschner1998,  Sotolongo-Costa2003, Galach2003a, Onofrio2005ab, Yafia2007, 
Rodriguez-Perez2007,forys2006} or \cite{Banerjee2007,d2010delay, Piotrowska2012RWA, Rihan2012,Bi2013, 
Rihan2014,Frascoli2014}.

In this article the influence of a 
particular modification to the generalization of the two-equation model proposed in \cite{forys2006} is studied. The proposed generalization concerns the form of the stimulus function, namely only some essential assumptions regarding the shape and regularity of this function are assumed for the theoretical part of the paper. The proposed modification reflects the fact that the immune system requires some time before the presence of a tumour (actually, tumour antigens)  in the body 
is detected and the immune system is able to respond. This assumption leads to the incorporation of a 
discrete time delay in the function describing the immune system stimulation by presence of the tumour cells. For the obtained system the basic mathematical properties of the model,  the possible 
asymptotic behaviour of the system depending on the delay and the role of the native background immunity and the immune stimulation strength are investigated.  For the particular form of the stimulus function the considered model is validated with two sets of the experimental data for particular tumour cell line  and further studied in the context of 
the model dynamics within the framework of the other experimental data.
It is shown that model can reflect different patterns, including periodic, reported in the experimental literature. 
For the estimated parameter values it is shown that an increase of the 
background immunity or the stimulation strength of the tumour cells on the immune system enlarge the stability region for the positive dormant steady state.
Moreover, the increase of the background immunity might stabilize the  
unrestricted 
tumour growth observed for the control experiment.

The paper is organized as follows: in Section \ref{MP} the model without delay and summary of the existence and 
stability of the non-negative steady states are presented; in Section \ref{DM} the model with discrete delay is 
described and 
analysed. 
Next, in Section \ref{sec:model_val} model is validated with the experimental data. The numerical simulations 
confirming and extending the analytical results are shown in Section \ref{sec:res}. Finally, in Section \ref{sec:diss}  
above results 
are summarised and discussed.
%
\section{Presentation of the base model}\label{MP}

In~\cite{forys2006} a simple model for tumour-immune system interaction, that takes into account the two most 
important variables:  the size of specific anti-tumour immunity at time $s$ ($X(s)$) and the size of  tumour
($Y(s)$), is studied. To  describe  the size of the overall specific immunity 
against the tumour Fory\'s et al.,  \cite{forys2006},  make the following assumptions.
First, in the absence of tumour antigens one observes a constant (with rate $w$) production of precursor 
cells that are able to respond to the tumour antigens and the natural death of the cells, which is expressed 
by a linear term $uX$. Thus, as a consequence of the tumour absence within the body a constant immunity 
level (so-called background immunity) is kept.  Second, the presence of tumour antigens stimulate an increase 
of the immunity level. Moreover, it is assumed that immunity response is proportional to the current size of the 
immunity $X$ with the proportionality coefficient ($aF$), which is assumed to be a bounded function of both 
immunity size and tumour antigens or tumour antigens only. It is also assumed that the creation of the 
antibody-antigen complexes lead to the death of the immune cells and it is modelled by a bilinear term $XY$ 
with constant coefficient $b$. Clearly, if for particular tumours this process is not observed, then $b=0$. Finally, 
the destructive influence of immune system on the tumour development is described by the bilinear term 
$XY$ with the constant rate $c$, similarly as it was assumed in \cite{Onofrio2005ab}. Additionally, it is 
assumed that the tumour grows according to the exponential law, as e.g. in~\cite{Sotolongo-Costa2003,Banerjee2007, Rodriguez-Perez2007}, with growth rate equal to $r$. 
The system proposed in\cite{forys2006} has the following form

\begin{subequations}\label{system:nondim}
	\begin{align}
		\frac{\dd X}{\dd s}&=w - uX + aF(X,Y)X - bXY, \label{system:nondim:1}\\
		\frac{\dd Y}{\dd s}&=Y(r-cX), \label{system:nondim:2} 
	\end{align}
\end{subequations}
where the presence of tumour antigens that stimulate the immune system response ($F$) is defined by  
\begin{align}
F_1(X,Y)	=\frac{Y^\alpha}{k_1^\alpha X^\alpha+Y^\alpha}, \quad \text{or}\quad
F_2(X,Y)	=\frac{Y^\alpha}{k_2^\alpha+Y^\alpha}\label{F1}.
\end{align}
Different forms of $F_1$ and $F_2$ reflect different assumptions on tumour--immunity interactions, and both 
of them were earlier considered in the literature in the context  of the immune system 
response. 
For  $F=F_2$ it is assumed that the number of stimulated cells does not depend on the size of the immune system itself, 
hence stimulation is signal dependent. Such a form of the stimulus function in the 
context of immune
 response was considered for example in~\cite{Bell1973, Kuznetsov1994, Mayer1995, Pillis2001} and more 
recently in~\cite{forys2006, d2010delay, Rihan2014}. In case of  $F=F_1$  a number of 
signal molecules need to be reached before faster production of immune cells is activated (after which, the
size of anti-tumour immunity increases) hence such a function is often called a ratio-dependent stimulation 
function and it was considered in the immune system - tumour interactions context in~\cite{Prikrylova1992}.

Due to the biological interpretation of constants, it is assumed that all parameters are non-negative and in 
particular parameters $a$, $c$, $u$, $r$, $w$ and $k_i$, $i=1,2$ are strictly positive. Moreover, it is assumed that 
$\alpha$, which describes the switching characteristics of 
the stimulus functions, is greater or equal to 1.  

In order to simplify the model analysis, in the analytical part of the paper, its non-dimensional form is considered

\begin{subequations}\label{system:original}
	\begin{align}
		\frac{\dd x}{\dd t}&=d\big(\omega -x+\rho f(x,y)x-\psi xy\big),\label{system:original:x}\\
		\frac{\dd y}{\dd t}&=y(1-x), \label{system:original:y}
	\end{align}
\end{subequations}
where $\omega=cw/(ur)>0$ is the native (background) immunity, i.e. the natural level of immunity in the absence of 
antigens, $\rho=a/u>0$ is the stimulation strength, $\psi=bY_0/u\geq0$ ($Y_0=k_1r/c$ or $Y_0=k_2$ 
depending on the form of $F$ function considered) is the strength of anti-immune activity, $d=u/r>0$ and 
$f(x,y)$ is the stimulus function given by

\begin{align}
f_1(x,y)=\frac{y^\alpha}{x^\alpha+y^\alpha}\quad\text{or}\quad
f_2(x,y)=\frac{y^\alpha}{1+y^\alpha},\quad \alpha\geq1,\label{f1}
\end{align}
time and variables are scaled $t=rs$ and $x=cX/r$, $y=Y/Y_0$, respectively.

As mentioned before in the well established literature
the stimulating effect of the presence of tumour antigens on the immune system
response is usually described by the function $f=f_2$ and often (\cite{Kuznetsov1994, Pillis2001,d2010delay, Rihan2014})
with $\alpha=1$ (that is by the Michaelis-Menten function). Hence, that form is considered in the part of paper devoted to the  model fitting to the experimental data and its further validation. However,  function $f_2$ given by~\eqref{f1} has several general properties, thus for the theoretical part on the paper a more general family of stimulus functions, fulfilling conditions: 
\begin{enumerate}[\bf ({A}1)]
	\item the function $f$ is a~non-negative Lipschitz continuous \label{as1}
		on $\R_{\ge}=[0,+\infty)$;
	\item $f(y)$ is bounded by one for all  $y\in\R_{\geq}$; \label{as2} 
	\item $f(0)=0$, $\lim\limits_{y\to+\infty} f(y) = 1$, $f(1)=1/2$;\label{as3}
	\item $f(y)$ is increasing for all $y\in\R_{\geq}$;\label{as4}
	\item $f(y)$ is differentiable for all $y\in\R_{\geq}$;\label{as5}
	\item the function $f$ has at most one inflection point. \label{as7}
\end{enumerate}
and if further needed the additional smoothness
\begin{enumerate}[\bf ({A}1)]
\setcounter{enumi}{6}
	\item $f(y)\in C^2$ for all $y\in\R_{\geq}$,\label{as8}
\end{enumerate}
is considered.

Clearly, assumption~(A\ref{as1}) together with the form of the system~\eqref{system:original} guarantees 
existence of non-negative unique solutions of~\eqref{system:original} with non-negative initial 
condition~\eqref{incond}. Additionally, the global existence of the solution is a~consequence of the assumption (A\ref{as2}).

In~\cite{forys2006} the importance of parameters $\omega$ (background immunity) and $\mu=\omega-1+\rho$, 
(so-called overall immune capacity) is stressed since both of them strongly influence the existence and the stability 
of the steady states. Hence, due to the biological interpretation of the parameter $\mu$ in the rest of the paper its positivity, i.e., 
\[\omega+\rho>1 \]
is assumed.

Because of the form 
of Eq.~\eqref{system:original:y} all non-negative steady states are always of the form $(\zeta,0)$ or $(1,
\zeta)$, where $\zeta>0$. Assumption $f(0)=0$ (see (A\ref{as3})) implies that there always exists 
a~semi-trivial steady state $A=(\omega,0)$. Clearly,  the existence and number of positive steady states depends 
on the number of solutions of equation
\begin{equation}\label{positiv:ss}
	\rho f(y)= \psi y+1 - \omega. 
\end{equation}

Assumption~(A\ref{as7}) implies that there exist at most three positive solutions of~\eqref{positiv:ss}. 
If $\psi=0$, then due to (A\ref{as3}) and (A\ref{as4}) there exist no positive solutions of~\eqref{positiv:ss} 
for $\omega>1$ and there exists exactly one $y_R$ for $\max\{0,1-\rho\}<\omega<1$. If $\psi>0$, then 
the situation is more complicated. If the function $f$ has no inflection point (the assumption (A\ref{as3}) 
implies that  $f$ is concave), then for $0<\omega<1$ Eq.~\eqref{positiv:ss} has zero or two solutions in 
a~generic case. 
On the other hand, for $\omega>1$ Eq.~\eqref{positiv:ss} has exactly one positive solution. If the function
$f$  has one inflection point, Eq.~\eqref{positiv:ss} has one or three positive solutions for $\omega>1$ and 
zero or two for $0<\omega<1$ (in a~generic case). 

Arguments analogous to those presented in~\cite{forys2006} imply that the stability of the steady states  of the 
system~\eqref{system:original} with $f$ fulfilling conditions (A\ref{as1})-(A\ref{as7}) are the same as for 
particular forms of $f$ functions given by~\eqref{f1}. 
The results regarding the existence and stability of the steady states of~\eqref{system:original} are 
summarized in Tab.~\ref{table1} and partially in Fig.~\ref{fig:0}. Note that, in~\cite{forys2006} additionally it 
is shown that there are no periodic solutions of the system~\eqref{system:original}.


\begin{center}
\begin{table}[htb]
\footnotesize{
\centering
\caption{Existence and stability of the steady states of the system \eqref{system:original} for  $f$ fulfilling conditions (A\ref{as1})--(A\ref{as7}) depending on values of parameters $\alpha,\omega,\psi,\rho$.  Numbers in
the column ips  indicate the number of possible inflection points of the function $f$.}\label{table1}
\begin{tabular}{|c| c| c| c|l |}\hline
$\psi$	&	$\omega$	& ips &	Steady states 	& Stability\\
\hline
\multicolumn{5}{|c|}{\textit{Semi-trivial steady states}}\\
\hline
\multirow{2}{*}{$\psi\ge0$}&	$0<\omega<1$	& \multirow{2}{*}{$0,1$}		& \multirow{2}{*}{$A=(\omega,0)$}					&	unstable\\
					\cline{2-2} \cline{5-5}	
					&	$\omega>1$	&	&	&locally stable\\
\hline
\multicolumn{5}{|c|}{\textit{Positive steady states}}\\
\hline
\multirow{3}{*}{$\psi=0$}	&	$0<\omega<1^*$	& \multirow{2}{*}{$0,1$}	&{$R=(1,y_R)$, $y_R>0$}	                        &     	{globally stable}\\
					\cline{2-2}\cline{4-5}
					&	$\omega>1$		&	& \multicolumn{2}{c|}{none}\\
\hline
\multirow{7}{*}{$\psi>0$}	& 	\multirow{3}{*}{$0<\omega<1$}	& \multirow{3}{*}{$0,1$}	&\multicolumn{2}{c|}{none}\\
												\cline{4-5}
					&				&			&$B=(1,y_B),\,\,C=(1,y_C)$ 				& $B$ locally stable\\
					&			&				&$0<y_B<y_C$			& $C$ unstable\\
					\cline{2-5}
					&\multirow{4}{*}{$\omega>1$}&	 $0,1$	&$C=(1,y_C)$, $y_C>0$				& $B$ unstable\\
												\cline{3-5}
					&					&	 $1$		&$B=(1,y_B),\,\,C=(1,y_C),\,\,D=(1,y_D)$			& $B$ unstable\\
					&					&		&$0<y_B<y_C<y_D$			& $C$ locally stable\\
					&						&	&			& $D$ unstable\\
\hline
\end{tabular}
\\
\small{$^*$ under assumption $\omega+\rho>1$}\\
}
\end{table}
\end{center}

Since in a later section of the paper $f$ is given by the Michaelis-Menten function, the conditions for the existence of positive steady states of the system~\eqref{system:original} for $\alpha=1$ and 
$f=f_2$ are directly provided below (compare with Figure 1).

\begin{prop}\label{prop1}
	Let $\psi=0$ and $f=f_2$ is given by~\eqref{f1}, then 
	\begin{enumerate}[\upshape (i)]
		\item for $\omega>1$ there are no positive steady states of the system~\eqref{system:original};
\item for $\max\{0,1-\rho\}<\omega<1$ there is exactly one positive steady state of the system~\eqref{system:original} given by
\[
R=(1,y_R)=\left(1,\Big(\frac{1-\omega}{\rho+\omega-1}\Big)^{\frac{1}{\alpha}}\right).
\]

	\end{enumerate}
\end{prop}
\begin{proof}
 Because of the form of Eq.~\eqref{system:original:y} all positive steady states are always of the form $(1,\zeta)$, where $\zeta>0$.
Clearly, the existence and number of positive steady states depends on the number of solutions of~\eqref{positiv:ss}. 
For $\psi=0$, due to the facts: $f(0)=0$,  $\lim\limits_{y\to+\infty} f(y) = 1$ and $f$ being increasing, there exist no positive solutions of \eqref{positiv:ss} for $\omega>1$ and there exists exactly one $\zeta$ for $\max\{0,1-\rho\}<\omega<1$.
\end{proof}

\begin{prop}\label{prop2}
	Let $\psi>0$, $\alpha=1$, $f=f_2$ is given by~\eqref{f1} and
		$\Delta = \left(1+(1-\omega-\rho)/\psi\right)^2 - 4(1-\omega)/\psi$,
 then 
	\begin{enumerate}[\upshape (i)]
		\item for $\omega>1$ there exits a~unique positive steady state $C=(1,y_C)$ of the system~\eqref{system:original} with 
		\begin{equation}\label{alpha1:yC}
			y_C = \frac{1}{2}\left(\frac{\rho+\omega-1}{\psi}-1+\sqrt{\Delta}\right).
		\end{equation}
		\item for $0<\omega<1$  and 
		$\rho >\bigl(\sqrt{\psi}+\sqrt{1-\omega}\bigr)^2$
		 two positive steady states $(1,y_B)$ and $(1,y_C)$ of the system~\eqref{system:original} exist where
		\[
			y_B = \frac{1}{2}\left(\frac{\rho+\omega-1}{\psi}-1-\sqrt{\Delta}\right)
		\]
		and $y_C$ is given by~\eqref{alpha1:yC}.
		On the other hand, if  $0<\omega<1$ and $\rho <\bigl(\sqrt{\psi}+\sqrt{1-\omega}\bigr)^2$, then 	there are no positive steady states of the system~\eqref{system:original}. 
	\end{enumerate}
\end{prop}
\begin{proof}

	Eq.~\eqref{positiv:ss} takes form
		$\rho y = \Bigl(\psi y+1-\omega\Bigr) \bigl(1+y\bigr)$.
	Dividing both sides by $\psi$ and denoting 
		$\beta = \frac{1-\omega}{\psi}$ and $\nu = \frac{\rho}{\psi}$
	one obtains 
	\begin{equation}\label{alpha1:rrss}
		y^2 + (1-\nu+\beta)y+\beta = 0.
	\end{equation}
	Inequality $\omega>1$ implies $\beta<0$, hence Eq.~\eqref{alpha1:rrss} has one positive solution and the part (i) follows. 

	For $0<\omega<1$ there is $\beta>0$. 
	If real roots of~\eqref{alpha1:rrss} exist, then they have the same sign. Moreover, they
	are positive if and only 
	\begin{equation}\label{war:dod}
		\nu >1+\beta. 
	\end{equation}
	On the other hand, the existence of real solutions to~\eqref{alpha1:rrss} depends 
	on the sign of the determinant
		$\Delta = (1-\nu+\beta)^2-4\beta = \nu^2-2\nu(1+\beta)+ (1-\beta)^2$.
	It can be easily calculated that $\Delta >0$ if and only if 
	\begin{equation}\label{niernamu}
		\nu < \bigl(1-\sqrt{\beta}\bigr)^2 \quad \text{ or } \quad 
		\nu > \bigl(1+\sqrt{\beta}\bigr)^2.
	\end{equation}
	Note, that for $\beta>0$ the first inequality of~\eqref{niernamu} is equivalent to
		$\nu < 1+\beta - 2\sqrt{\beta}$,
	which contradicts~\eqref{war:dod}. On the other hand, the second inequality of~\eqref{niernamu}
	is equivalent (again for $\beta>0$) to
		$\nu > 1+\beta+2\sqrt{\beta}$,
	which implies~\eqref{war:dod}. Thus, the part (ii) follows. 
\end{proof}

%

\begin{figure}[htb!]
\centerline{
\includegraphics[width=.8\textwidth]{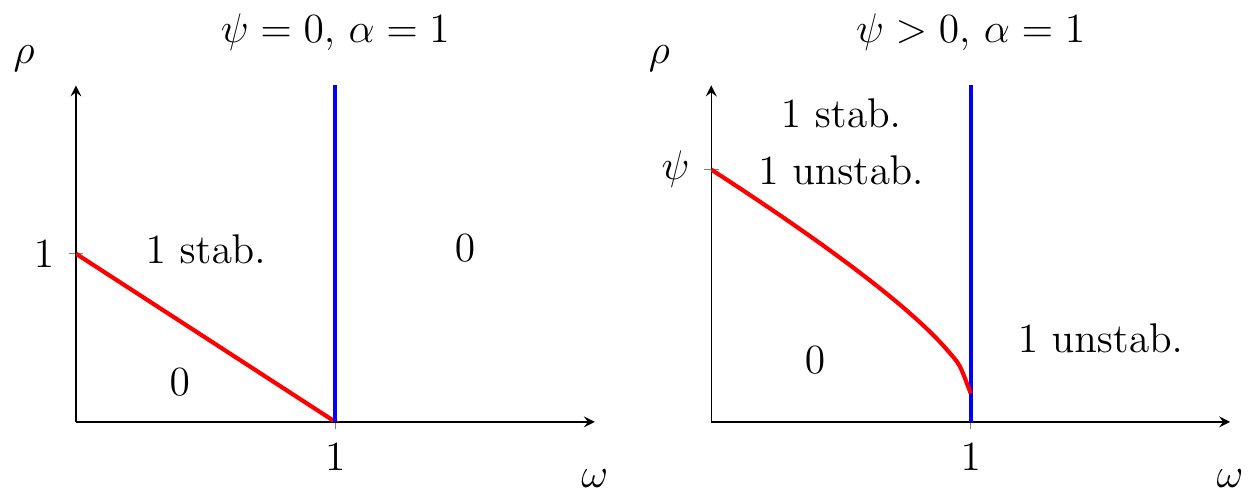}}
\caption{The $\omega-\rho$ space divided into regions where  different numbers of steady states for $\alpha=1$ 
and $\psi=0$ or $\psi>0$ are observed. Numbers and text in the regions indicate the number of positive steady states of the
system~\eqref{system:original} for $f=f_2$ and theirs stability for $\tau=0$; 
abbreviations stab. and unstab. stand for stable and unstable, respectively.}\label{fig:0}
\end{figure}

%
\section{Model with discrete time delay}\label{DM}

The immune system requires some time before the presence 
presence of the tumour (or rather, the tumour's antigens) in the organism
 is detected and the immune system is able to respond. From 
the mathematical point of view, to reflect such phenomena, one can include 
 time lags into the system, e.g. by introducing a discrete delay into the stimulus function. Although the model itself 
differs from the one considered in~\cite{d2010delay}, the idea presented there of  introducing one
discrete time delay only in the function 
describing the stimulating effect of the tumour on immune cells proliferation in the system variable representing 
the tumour cells compartment is followed here. That assumption reflects situation when the immune system 
response due to the presence of the tumour cells 
depends on 
the amount of the immune system cells at a given time $t$ and on the amount of tumour cells at 
time $t-\tau$, where $\tau>0$ 
is the time needed for the immune system to develop a suitable response after antigens recognition.
Following that approach time lag is introduced in a classical way, as it is done in the  logistic model by 
Hutchinson \cite{hutchinson48} or in the Gompertz model \cite{Piotrowska2011a,Bodnar_MBE_gomp2013},
where delay is introduced to the function representing the per capita growth rate. The same assumption of the 
form of time delay in the model of the immune system-tumour interactions has been recently made e.g. 
in~\cite{Bi2013}.
Clearly, in the literature~\cite{Galach2003a,Yafia2007,Rodriguez-Perez2007} discrete delay in the variable 
representing the size of the immune system was also considered. However, this idea is not followed since the 
priority is given to the immune system's time reaction. It is also assumed that the immune system does not 
require time to know the size of itself. 

The non-dimensional system with one discrete delay has the following form
\begin{subequations}\label{system:onedelay}
	\begin{align}
		\frac{\dd x(t)}{\dd t}&=d\big(\omega -x(t)+\rho f(y(t-\tau))x(t)-\psi x(t)y(t) \big), \label{system:onedelay:1}\\%
		\frac{\dd y(t)}{\dd t}&=y(t)(1-x(t)),\label{system:onedelay:2}
	\end{align}
\end{subequations}
where function $f$ fulfils conditions~(A\ref{as1})-(A\ref{as7}) (and (A\ref{as8}) if needed).

To close system~\eqref{system:onedelay} one needs to define the initial condition
\begin{equation}\label{incond}
x(t)=\varphi_1(t),\quad y(t)=\varphi_2(t),\,\,\text{for } t\in [-\tau,0], \quad\varphi_1(t), \varphi_2(t)\in\Cb,
\end{equation}
where $\Cb= \mathbf{C}\left([-\tau,0],\mathbb{R}_{\ge}\right)$ denotes the set of continuous functions defined on the interval $[-\tau,0]$ with non-negative values.
Clearly, since in system~\eqref{system:onedelay} only $y$ variable has delayed argument model dynamics will depend on $\varphi_1(0)$ and $\varphi_2(t)$.

\begin{theorem}\label{existence}
Let $f$ fulfils assumptions~(A\ref{as1})-(A\ref{as2}), then there exists a global, unique and non-negative solution of system~\eqref{system:onedelay}-\eqref{incond}.
\end{theorem}


\begin{proof}
For $t\in[0,\tau]$ $y(t-\tau)=\varphi_2(t-\tau)\geq0$ is well defined thus 
for $\varphi_1(t),\varphi_2(t)\in\Cb$, under assumption~(A\ref{as1}), there exists an interval $[0,t^*)\subset[0,\tau]$ where the solution $(x(t),y(t))$ of the system~~\eqref{system:onedelay} exists and is unique, for details see~\cite{Hale1993}.

Re-writing Eq.~\eqref{system:onedelay:2} in the integral form
($y(t)=\varphi_2(0)\exp{\int_0^t(1-x(s))ds}$)
one obtains non-negativity of $y(t)$ whenever it exists. Now, assume that 
at point  $t_1\in(0,t^*)$  $x(t)$ starts to be negative, i.e.  $x(t)>0$ for $0<t<t_1$, $x(t_1)=0$ and $\dot{x}(t_1)\leq0$.
Clearly, for considered $f(y)$ inequality
$\dot{x}(t_1)	=d\big(\omega-x(t_1)+\rho f(y(t_1-\tau))x(t_1)-\psi x(t_1)y(t_1)\big)
				=d\omega >0$, holds.
That contradicts the hypothesis that solution $x(t)$ starts to be negative for $t=t_1$.

On the other hand, the non-negativity of $x(t)$ and $y(t)$ imply $y(t)\le\varphi_2(0)e^{t^*}$. Moreover, stimulus function $f(y)$ fulfils assumptions~(A\ref{as1})-(A\ref{as2}). Hence, $\dot{x}\le d(\omega+\rho x)$ and $x(t)\le e^{d\rho t^*}\left(\varphi_1(0)+\frac{\omega}{\rho}\right)-\frac{\omega}{\rho}$. Thus, both 
variables and theirs derivatives are bounded, so the solution can be extend through the point $t^*$. Following the method of steps (see 
e.g.~\cite{Hale1993}) and repeating presented arguments for each interval $[n\tau, (n+1)\tau]$, 
$n\in\mathbb{N}$ one completes the proof.
\qquad\end{proof}

%
\subsection{Stability of the steady states}
Clearly, the steady states of the system with delay~\eqref{system:onedelay} are the same as for the system 
without delay~\eqref{system:original}. However, the system's stability, already strongly depending on the 
values of parameters $\phi$ and $\omega$ without delay, now can also depend on the value of parameter 
$\tau$.

\begin{theorem}\label{trivialss}
Let $f$ fulfils conditions (A\ref{as1})--(A\ref{as5}). The semi-trivial steady state $A=(\omega, 0)$ of system~\eqref{system:onedelay}-\eqref{incond} is a locally stable node for $\omega>1$ and a saddle for $\omega<1$ independently of the value of parameter $\tau\geq0$.
\end{theorem}
\begin{proof}
For $\tau=0$ steady state $A$ is a locally stable node (for $\omega >1$) or a saddle (for $0<\omega<1$) and it 
always exists, for details see~\cite{forys2006}.
Consider $\tau>0$. Linearising system~\eqref{system:onedelay} around steady state $A$ one obtains system 
without delay which is the same as the linearisation of the system~\eqref{system:original} without delay. Hence, parameter $\tau$ has no 
influence on the local stability of semi-trivial steady state $A$. 
\qquad\end{proof}
As pointed earlier, for all positive steady states the first coordinate equals to 1, thus  
$(1,\zeta)$, $\zeta>0$ denotes the positive steady state of the system~\eqref{system:onedelay}.
Moreover, to simplify the notation define
\begin{equation}\label{betagamma}
	 \gamma_f :=  \rho f' (\zeta)=\rho\frac{\alpha\zeta^{\alpha-1}}{(1+\zeta^{\alpha})^2}>0,
\end{equation}
which is positive due to the assumption {\bf(A\ref{as4})}.
\begin{theorem}\label{pmn}
Let $f$ fulfils conditions (A\ref{as1})--(A\ref{as8}), $\psi\geq0$ and $\gamma_f>0$ is given by~\eqref{betagamma}.
\begin{enumerate}
\item If $\psi-\gamma_f>0$, then the steady state $(1,\zeta)$ of system~\eqref{system:onedelay} is unstable for all $\tau\ge0$.\label{aaa}
\item If $\psi-\gamma_f<0$, then the steady state $(1,\zeta)$ of system~\eqref{system:onedelay} is locally asymptotically stable for $\tau=0$ and there exists a~critical value $\tau_{cr}^\zeta$ such that $(1,\zeta)$ is locally asymptotically stable for $\tau<\tau_{cr}^\zeta$ and is unstable for $\tau>\tau_{cr}^\zeta$. Moreover, for $\tau=\tau_{cr}^\zeta$, the Hopf bifurcation occurs and periodic solutions arise with period equal to $2\pi/\sqrt{\eta^\zeta}$ , where
\begin{equation}\label{criticaldelay2}
\eta^\zeta=\frac{-(2d\zeta\psi+(d\omega)^2)+\sqrt{(d\omega)^4+4d^3\zeta\psi \omega^2+4(d\zeta\gamma_f)^2}}{2}, 
\quad
\tau_{cr}^\zeta=\frac{1}{\sqrt{\eta^\zeta}}\arccos{\left(\frac{\eta^\zeta+d\zeta\psi}{d\zeta\gamma_f}\right)}.
\end{equation}
\label{bbb}
\end{enumerate}
\end{theorem}

\begin{proof}
Linearising system~\eqref{system:onedelay} around steady state $(1,\zeta)$ one obtains
\begin{equation*}
		\frac{\dd x(t)}{\dd t}=d(-1+\rho f(\zeta)-\psi)x(t) -d\psi y(t) +\rho\gamma_f y(t-\tau),\quad\quad
		\frac{\dd y(t)}{\dd t}=-\zeta x(t).
\end{equation*}
Note that for  steady state $(1,\zeta)$ equality $1-\rho f(\zeta)+\psi\zeta=\omega$, holds. Thus,  the characteristic function for each positive steady state $(1,\zeta)$ has the following form
\begin{equation}\label{cqp}
W(\lambda) =\lambda^2 +d\omega\lambda - d\psi\zeta +d\zeta\gamma_f\e^{-\lambda\tau}.
\end{equation}
Clearly, for $\tau=0$ Eq.~\eqref{cqp} reads $W(\lambda)=\lambda^2+d\omega\lambda-
d\zeta(\psi-\gamma_f)$ and for $\psi-\gamma_f>0$ it has two real solutions of opposite signs. Thus, in such a 
case steady state $(1,\zeta)$ is a saddle. On the other hand, for $\tau=0$ and $\psi-\gamma_f<0$ 
both roots are real and negative or complex with negative real parts, hence $(1,\zeta)$ is locally asymptotically 
stable.
From the continuous dependence, one knows that for sufficiency small delays this result holds, however, the 
stability of the steady state(s) might change for a larger values of delay, see e.g.~\cite{Hale1993}. 
Such a situation is only possible if there exists a critical value $\tau=\tau_{cr}^{\zeta}$ and  a pair of purely imaginary roots of~\eqref{cqp} such that roots cross the imaginary axis (in the proper direction)
when $\tau$ exceeds $\tau_{cr}^{\zeta}$.

Let $\pm i\theta$ be a pair of pure imaginary roots of~\eqref{cqp}. Hence, 
$W(i\theta)=-\theta^2+id\omega\theta-d\zeta\psi+d\zeta\gamma_f\big(\cos{\theta\tau}-i\sin{\theta\tau}\big)=0$,
and
\begin{equation}\label{tetataupsi}
\cos{\theta\tau}=\frac{\theta^2+d\psi\zeta}{d\zeta\gamma_f}>0,\quad
\sin{\theta\tau}=\frac{d\omega\theta}{d\zeta\gamma_f}>0.
\end{equation}
A solution of system~\eqref{tetataupsi} exists if condition $d\zeta(\psi-\gamma_f)\leq-\theta^2$ is fulfilled. Hence, 
for $\psi-\gamma_f>0$ the characteristic quasi-polynomial~\eqref{cqp} of the system~\eqref{system:onedelay}  has 
no purely imaginary roots and hence the real parts of eigenvalues can not 
change sign for any value of $\tau\ge0$. 
That completes the proof of part~\ref{aaa}.

%

Consider $\psi-\gamma_f<0$. Squaring and adding  Eqs.~\eqref{tetataupsi} one has
\begin{equation}\label{2:19}
\tilde F(\theta):=\theta^4 
+d\Bigl(2\psi\zeta+d\omega^2\Bigr)\theta^2+d^2\zeta^2\Bigl(\psi^2-\gamma_f^2\Bigr)=0.
\end{equation}
Next, substituting $\eta=\theta^2$ one gets
$\eta^2+\eta((d\omega)^2+2d\zeta\psi)+(d\zeta)^2(\psi^2-\gamma_f^2)=0$
 with roots given by
\begin{equation}\label{eta}
\eta_{1,2}^\zeta=\frac{-(2d\zeta\psi+(d\omega)^2)\pm\sqrt{(2d\zeta\psi+(d\omega)^2)^2-4(d\zeta)^2(\psi^2-\gamma_f^2)}}{2}.
\end{equation}
Clearly,  roots~\eqref{eta} are always real since the expression $(2d\psi\zeta+(d\omega)^2)^2-4(d\zeta)^2(\psi^2-\gamma_f^2)$ is strictly positive for $\psi-\gamma_f<0$. Moreover, for $\psi-\gamma_f<0$ roots have opposite signs. 
Denote the positive root by $\eta_1^\zeta>0$, then there exists $\bar\theta^\zeta=\sqrt{\eta_1^\zeta}$ such that $0<(\bar\theta^\zeta)^2+2d\zeta\psi<d\zeta\gamma_f$, and moreover $W(i\bar\theta^\zeta)=0$ for an infinite sequence of critical values of $\tau$ given by
\begin{equation}\label{tauk}
\tau_{cr,k}^\zeta=\frac{1}{\bar\theta^\zeta}\left(\arccos{\left(\frac{\left({\bar{\theta^\zeta}}\right)^2+2d\zeta\psi}{d\zeta\gamma_f}\right)}+2k\pi\right), \quad k\in\mathbb{N}_0.
\end{equation}

Denote by $\tau_{cr}^\zeta$ the first positive delay fulfilling condition~\eqref{tauk}.
To complete the proof one needs to show that the real part of eigenvalues of roots cross the imaginary axis at $\tau=\tau_{cr}^{\zeta}$ with positive speed, i.e. $\frac{\partial}{\partial\tau}\Re {\lambda(\tau)}|_{\tau=\tau_{cr}^{\zeta}}>0$ and that the re-stabilization of the positive steady state is not possible for larger values of $\tau$. Here, using the theorem proved by Cooke and Van den Driessche in~\cite{Cooke1986}, one can show that  all the roots of the characteristic quasi-polynomial~\eqref{cqp} cross the imaginary axis from the left to the right-hand side of the complex plane. 
The Cooke and Van den Driessche theorem indicates that $\frac{\partial}{\partial\tau}\Re {\lambda(\tau)}|_{\tau=\tau_{cr,k}^{\zeta}}=\frac{d\widetilde{F}(\theta)}
{d\theta}|_{\eta=\bar\eta^{cr}}$,
hence it is enough to calculate $\frac{d\widetilde{F}(\theta)}
{d\theta}=4\theta^3 +2d\theta(2\psi\zeta+d\omega^2)>0$ for $\theta>0$.
Thus, in the case when the steady state is locally asymptotically stable for $\tau=0$, for $\tau=\tau_{cr}^\zeta$  
new periodic orbits arise due to the Hopf bifurcation with period equal to 
$2\pi/\bar{\theta^\zeta}$. Moreover, destabilized steady  state $(1,\zeta)$  remains unstable for all 
$\tau>\tau_{cr}^{\zeta}$. That completes the proof. \qquad
\end{proof}

%
\section{Experimental data, estimation of parameters and model validation}\label{sec:model_val}

The model parameters are obtained by fitting the solutions of the non-scaled version 
of model~\eqref{system:onedelay} to the experimental data provided in~\cite{Uhr1991} for mice affected by 
B-cell lymphoma in the spleen (BCL$_1$). It is assumed that $F$ is given by the Michaelis-Menten 
function ($F=F_2$ defined by~\eqref{F1} with $\alpha=1$) since it is the most common form of the 
stimulation function considered in that context in the literature, see e.g. \cite{Kuznetsov1994, 
Pillis2001,d2010delay, Rihan2014}. Estimation of the ranges for each parameter is based on the experimental 
literature as well as on the assumptions made by other scientists who in particular modelled  the BCL$_1$ and 
DLBCL  (diffuse large B-cell lymphoma) progression or modelled the tumour-immune system interactions in 
general \cite{Kuznetsov1994,Rihan2014,Roesch2013,Pillis2001,Rodriguez-Perez2007,Kirschner1998}.

%
\subsection{Experimental data}
The solutions to non-scaled version of model~\eqref{system:onedelay} are fitted to the experimental data for BCL$_1$ from~\cite{Uhr1991}, where the modelled experiment was performed on two sets of animals: the control group (BALB/c mice) and chimeric. Chimeric animals were obtained from BALB/c mice by having, by  lethal irradiation, their  bone  marrow  destroyed and then having bone marrow  transferred from other mice 
to rebuild the immune system. In  normal  BALB/c mice,  one observes that $5\times10^5$ injected BCL$_1$ cells into the spleens grew rapidly, reaching a plateau after 40 days and causing mice death by 90 days.
In contrast, identically injected groups of chimeric animals demonstrated a different growth pattern: the initial 
rate of growth was similar to the control group,
but then it slowed down with tumours reaching their maximum sizes (10 times smaller than in control mice) and a commensurate decrease in the number of cells was observed after day 40, see Figure~\ref{fig:2}.
%
\subsection{Parameter ranges}

{\bf Tumour growth rate ($r$).} Tumour volume duplication time strongly depends on the tumour type. For some tumours it can be 20 days, while for others, e.g. colon cancer, it can be years. 
On the other hand, the tumour doubling time for malignant lymphoma is reported to be around 29
days, \cite{Tubiana1989}. In \cite{Roesch2013}, for DLBCL, 
it was assumed that the tumour doubling time is between 1.4 and 70 days which corresponds to a tumour growth rate between 0.0099 and 0.4951 day$^{-1}$ (approximately 0.01 and 0.5 day$^{-1}$). That assumption is followed since the interval [0.01,0.5] includes the range 0.05-0.5 proposed in \cite{Rihan2014}.
Clearly, the tumour growth rate estimated by Kuznetsov et al. in \cite{Kuznetsov1994} (0.18 day$^{-1}$) also lies in that interval.




{\bf  Natural lymphocyte death rate ($u$).}
The mean lifetime of  T lymphocytes  from  the  spleen  and  the 
blood was estimated to be approximately 30 days 
or more,~\cite{Reynolds1985}. 
Thus, the degradation rate, which is 1/(mean life time),  is assumed to be in the range $1.25\times10^{-2}-5\times10^{-2}$ day$^{-1}$, which corresponds to $20-80$ days mean life time. The values of the death rate $0.03125$  and $0.0412$ day$^{-1}$ estimated in \cite{Rihan2014} and \cite{Kuznetsov1994}, respectively, also lie in the assumed range. 

{\bf Immune response function ($F$,  $\alpha$, $k_2$, $a$).} It is assumed that the immune system response term 
has the same form as those considered e.g. in~\cite{Kuznetsov1994,Kirschner1998,Pillis2001,d2010delay,Rihan2014}  or \cite{Roesch2013}, that is $F=Y^{\alpha}/(k_2^{\alpha}+Y^{\alpha})$ with 
$\alpha=1$. In \cite{Kuznetsov1994} the value of  the steepens of the immune response function ($k_2$) was 
estimated to be $2.019\times10^7$ cells, while in \cite{Rihan2014} the value $1.5109\times10^7$ cells.
In this study the range of $10^7-2\times10^8$ cells for $k_2$ is chosen.
To determine the range of the maximal immune system response rate ($a$)  \cite{Rihan2014} is followed and 
range 0-2.5 day$^{-1}$, which contains the value $0.1245$ day$^{-1}$ estimated in~\cite{Kuznetsov1994}, is 
taken.

{\bf  Tumour independent production rate of effector cells ($w$).}
Following~\cite{Rihan2014} the range for the constant influx of effector cells ($w$) from the immune 
system in the absence of tumour is chosen to be $3.2\times10^3-3.2\times10^4$ cells day$^{-1}$. The value 
$1.3\times10^4$ cells day$^{-1}$  estimated by \cite{Kuznetsov1994} belongs to the considered range.

{\bf Tumour-induced inactivation rate of effector cells ($b$). }
In~\cite{Rihan2014} the range for the fraction of immune cells inactivated due to the interactions with tumour cells 
($w$) was assumed to be $10^{-12}-5\times10^{-7}$ cells$^{-1}$ day$^{-1}$. The value $3.422\times10^{-10}$ 
cells$^{-1}$ day$^{-1}$ estimated in~\cite{Kuznetsov1994} belongs to the assumed range. 

{\bf Effector induced elimination rate of tumour cells (c).}
In~\cite{Rodriguez-Perez2007} the range for the fraction of cancer cells killed per effector cell 
was assumed to be  $10^{-9}-10^{-7}$ cells$^{-1}$day$^{-1}$. However, in \cite{Kuznetsov1994} the value of that parameter was estimated to be equal to $1.101\times10^{-7}$  cells$^{-1}$day$^{-1}$. Hence, the  range is extended and it is assumed to be $c\in[10^{-9},3\times10^{-7}]$. 


{\bf Time delay in the immune system activation ($\bar\tau$)}.
In this study it is assumed that time delay in immune system response is in range from 1h till 5 days, i.e. 0.0416-5 days. 

All ranges of the model parameters together with corresponding references are given in Table~\ref{tab:range}. 

\begin{table}[htb]
\caption{Parameter ranges with referenced used for the fitting procedure to the experimental data from~\cite{Uhr1991}.}\label{tab:range}
\begin{tabular}{|lp{0.33\linewidth}llp{0.16\linewidth}|}
\hline
Par. & Biological interpretation & Unit & Range& Reference \\
\hline
$w$ &  tumour independent production rate of effector cells & cells day$^{-1}$ & $3.2\times10^3-3.2\times10^4$ & \cite{Rihan2014}, \cite{Kuznetsov1994} \\
\hline
$u$ &  natural lymphocyte death rate & day$^{-1}$ & $1.25\times10^{-2}-5\times10^{-2}$ &\cite{Reynolds1985} \cite{Rihan2014}, \cite{Kuznetsov1994}.  \\
\hline
$a$ & maximal immune response rate   & day$^{-1}$ & 0-2.5 & \cite{Rihan2014}, \cite{Kuznetsov1994} \\ 
\hline
$b$ & fraction of immune cells inactivated in interactions with tumour cells  &  cells$^{-1}$ day$^{-1}$ & $10^{-12}-5\times10^{-7}$ & \cite{Rihan2014}, \cite{Kuznetsov1994}\\
\hline
$r$ & tumour growth rate &  day$^{-1}$ & $0.01-0.5$ &  \cite{Roesch2013}, \cite{Rihan2014}, \cite{Kuznetsov1994}, \cite{Tubiana1989} \\
\hline
$c$ & fraction of cancer cells killed per effector cell & cell$^{-1}$ day$^{-1}$ & $10^{-9}-3\times10^{-7}$ & \cite{Rodriguez-Perez2007}, \cite{Kuznetsov1994}
\\
\hline
$\bar\tau$ & time delay & day & $ 0.0416-5$ & ass.  \\
\hline
$k_2$ & steepens of immune response  & cells & $10^7-2\times10^8$ & \cite{Rihan2014},\cite{Kuznetsov1994} \\
\hline
$\alpha$&detrmines type of immune response &--&1& \cite{Kuznetsov1994}, \cite{Kirschner1998}, \cite{Pillis2001}, 
\cite{d2010delay}, \cite{Rihan2014}, \cite{Roesch2013}\\
\hline
\end{tabular}
\end{table}


\subsection{Estimation of model parameters}
Clearly, if there are no tumour cells in the system ($Y=0$), then immune-activity stays at the constant level 
calculated directly from~\eqref{system:nondim:1}. Hence, the initial condition for the variable $X$ is a constant equal to 
$w/u$. On the other hand, to determine the level of the tumour cells at the 
beginning of the experiment (at $s=0$) one can directly use the data from \cite{Uhr1991} taking $\Phi_2(0)=5\times10^5$ 
cells. Thus, the initial condition for non-scaled  version of \eqref{system:onedelay} reads 
\begin{equation}\label{eq:initcond}
\Phi_1(s)=\frac{w}{u} \,\,\,\textrm{for}\,\,\, s\in[-\bar\tau,0], \quad \textrm{and} \quad  \Phi_2(s)=\begin{cases}
					0 & \textrm{for}\,\,\, s\in[-\bar\tau,0), \\
					 5\times10^5& \textrm{for}\,\,\, s=0, 
				\end{cases} 
\end{equation}
and due to the dependencies on the parameters $w$, $u$ and $\bar\tau=\tau/r$ the initial condition is estimated within the fitting procedure. 

The fitting procedure is performed in the commercial MATLAB$^\circledR{}$ computational environment using the
{\it PSO} function. PSO finds a minimum of a function of several variables using the particle 
swarm optimization algorithm based on the stochastic optimization technique that was inspired by the social behaviour 
of bird flocking or fish schooling. The algorithm was originally introduced in 1995 by Kennedy and Eberhart, 
\cite{Kennedy1995}, and later extended by Shi and Eberhart, \cite{Shi1998}. Nowadays it has a number 
of extensions, however the main idea is the same a swarm of particles moves around in the search space and 
the movements of the individual particles are influenced by the improvements discovered by the others in the 
swarm.  In~\cite{Clerc2002} Clerc and Kennedy introduced a constriction factor in PSO, which was 
later shown to be superior to the inertia factors. That version of algorithm 
is used here. Solutions of delay differential equations (DDEs) are computed using {\it dde23} function with tolerance RelTol =1e-7, and 
AbsTol=1e-11.


In the presented study the objective function is the standard one based on the calculation of mean square error (MSE):
\[
\textrm{MSE}=\frac{1}{m}\sum_{i=1}^{m}err_i^2, \quad \quad err_i=\frac{x^{data}_i-x^{sym}_i}{x^{data}_i},
\]
where $m$ is the number of experimental measurements, $x^{data}_i$ are the measured values and $x^{sym}_i$ are the values of the solution of DDEs for respective time points.

\begin{table}[htb]
\caption{Fitted non-scaled model parameters, steady states and theirs stability for $\bar\tau=0$ for non-scaled system~\eqref{system:onedelay}. Parameter $\alpha$=1.}\label{table2}
\centering
\begin{tabular}{|>{$}l<{$}*{3}{|>{$}r<{$}}}\hline
& \multicolumn{1}{c|}{control} & \multicolumn{1}{c|}{chimeric}  \\ \hline
\multicolumn{3}{|c|}{\it parameters}\\ \hline
w  & \multicolumn{2}{c|}{$3.0337\times10^4$} \\[2pt]
u  & \multicolumn{2}{c|}{$0.0196$}  \\[2pt]
a & 0.1452 & 0.2835  \\
k_2 & 1.3604\times 10^8 & 1.8720\times10^8 \\
b  & 2.1010\times10^{-10} & 6.6828\times10^{-10} \\
r & \multicolumn{2}{c|}{0.4695}  \\[2pt]
c & 1.7608\times10^{-7} & 2.0761\times10^{-7} \\
\bar\tau & 3.1869 & 0.0581 \\ \hline
\multicolumn{3}{|c|}{\it steady states and theirs stability for $\bar\tau=0$}\\ \hline
\text{semi trivial unstable}  & (1.5472\times 10^6, 0) & (1.5472\times10^6, 0) \\
\text{positive stable} & (2.6666\times 10^6, 1.0548\times 10^7)  & (2.2616\times 10^6, 7.8919\times 10^6)  \\
\text{positive unstable} &(2.6666\times 10^6,5.0530\times 10^8) & (2.2616\times10^6, 2.1986\times10^8) \\
\hline
\end{tabular}\\
\end{table}


\begin{table}[htb]
\caption{Fitted model parameters for system~\eqref{system:onedelay}, 
parameter $\alpha$=1.}\label{table3}
\centering
\begin{tabular}{|c|*{5}{r}|}\hline
& $d$ & $\omega$ & $\rho$ & $\psi$ & $\tau$ \\
\hline
control & \multirow{2}{*}{$0.0418$} & $0.5802$ & $7.4046$ & $1.4576$  & $1.4964$ \\
chimeric & & $0.6841$ & $14.4583$ & $6.3802$ & $0.0273$ \\
\hline
\end{tabular}\\
\end{table}

\begin{table}[htb]
\caption{Summary of the bifurcation phenomena for scaled parameters given in~Table~\ref{table3} and $\tau$ treated as a bifurcation parameter, $\alpha=1$.}\label{table4}
\centering
\begin{tabular}{|c|l|c|c|c|c|}\hline
Param. type & steady state & stability $\tau=0$ & $\tau_{cr}$ & period & bif. type\\
\hline
\multirow{3}{*}{control} & (0.5802, 0) & unstable &  \multicolumn{3}{c|}{unstable for $\tau\geq0$} \\
\cline{2-6} 
				    & (1, 0.0775) & stable & 1.1776 & 50.1363& supercritical\\
\cline{2-6}
				    & (1, 3.7145) & unstable &  \multicolumn{3}{c|}{unstable for $\tau\geq0$} \\
\hline
\multirow{3}{*}{chimera} & (0.6841, 0) & unstable &  \multicolumn{3}{c|}{unstable for $\tau\geq0$} \\
\cline{2-6}
				    & (1, 0.0422) & stable & 1.2227 & 57.3715& supercritical\\
\cline{2-6}
				    & (1, 1.1745) & unstable &  \multicolumn{3}{c|}{unstable for $\tau\geq0$} \\
\hline
\end{tabular}\\
\end{table}

First, the solutions of the non-scaled version of system~\eqref{system:onedelay} are fitted to the data from the 
chimera experiment for BCL$_1$ tumour cell population in the spleen of the BALB/c mice over 110 days (described 
in~\cite{Uhr1991}) with the initial 
condition given by~\eqref{eq:initcond} that depends on $w$, $u$ and $\bar\tau$ parameter values. The result of 
the 
fitting procedure is presented in Figure~\ref{fig:2} (left) showing a good agreement with the experimental data (with MSE at the level of $3.57\%$) and thus positively validates the proposed model.

Next,  the parameters $w$, $u$ and $r$ are fixed (as for chimera experiment) and the rest of the parameters are fitted to the 90 day control mice experiment to further validate the model. This time  a very~good agreement with the experimental data 
with MSE at the level of $0.14\%$ is obtained, see Figure~\ref{fig:2} (right). The corresponding values of the estimated parameters for both experiments are given in Table~\ref{table2}.

\begin{figure}[htb!]
\centerline{\includegraphics[width=0.95\textwidth]{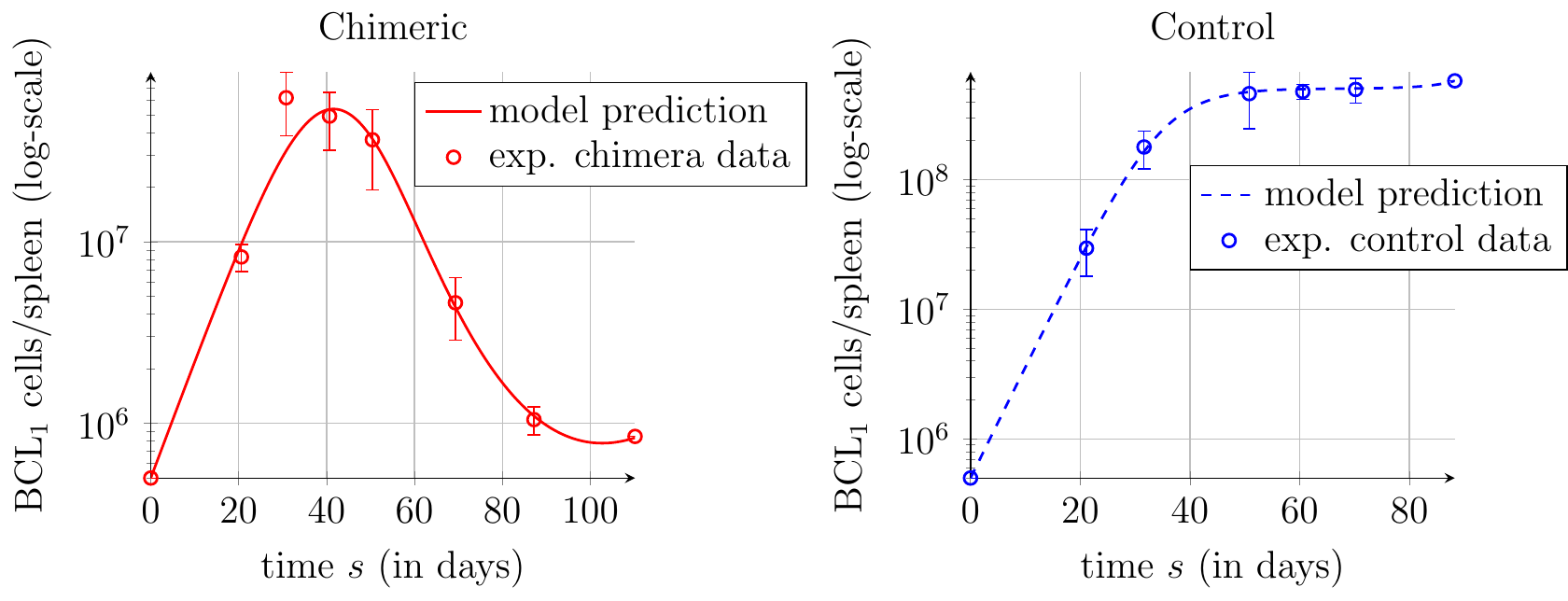}}
\caption{(Left) Result of the fitting procedure of the trajectories of the non-scaled version of system~\eqref{system:onedelay}  to the experimental data for BCL$_1$ tumour cell population in the spleen of the chimeric mice over 110 days experiment described in~\cite{Uhr1991}. MSE=0.0357.
(Right) Result of the fitting procedure of the trajectories of the non-scaled version of system~\eqref{system:onedelay} to the experimental data for BCL$_1$ tumour cell population in the spleen of the BALB/c mice over 90 days control experiment described in~\cite{Uhr1991}. MSE=0.0014. In both panels solutions of the non-scaled version of system~\eqref{system:onedelay} are computed with RelTol=1e-7 and AbsTol=1e-11. All used parameter values, for both panels, are given in Table~\ref{table2}. Bars indicate the measurement errors as reported in~\cite{Uhr1991}.}\label{fig:2}
\end{figure}
%
\section{Results}\label{sec:res}

Comparing the sets of parameters (see Table~\ref{table2}) for both experiments estimated via the fitting procedures one 
immediately sees that for the chimeric set of parameters parameter $a$ is two times larger than for the 
control set. Thus, together with slightly larger $k_2$ (for the chimeric experiment) implies a stronger stimulation of the immune 
system due to the presence of the tumour cells. 
Additionally, the time lag of that response is much shorter, approx. 1.4 h, for 
the chimeric set of parameters than for the control, approx. 76.5~h, indicating a more effective immune system reaction in the presence of the tumour cells for the chimeric mice.
For both the chimeric and control sets, the estimated value of the tumour cells killed per effector cells $(c)$ is very similar.
Despite the large range for parameter $b$ ($10^{-12}-5\times10^{-7}$, see Table~\ref{tab:range}) and 
that no further assumptions on the range of $b$ to differentiate the chimeric and control experiments 
are made, the fitted values obtain the same order of magnitude.


The simulations presented in Figure~\ref{fig:3} again validate positively the model and confirm results from~\cite{Uhr1991} showing 
that for larger than control initial injection of populations of the tumour BCL$_1$ cells into the spleens of the chimeric 
mice leads to the same growth pattern. Experiments performed in~\cite{Uhr1991} indicate that even if one injects 
$5\times10^7$ cells after around 110 days, due to the immune system response, the number of tumour cells is 
reduced to the almost the same amount as for $5\times10^5$ initially injected cells. To reproduce that 
phenomena all used parameter values were fixed, as given in Table~\ref{table3} (column chimeric) and only a 
part of the initial condition that is $\Phi_2(0)$ was changed, see Figure~\ref{fig:3}.
From the mathematical point of view such behaviour shows that the stable steady state, that can be identified as a dormant tumour state, has large enough basin of attraction.  

\begin{figure}[htb!]
\centering
\includegraphics[width=0.49\textwidth]{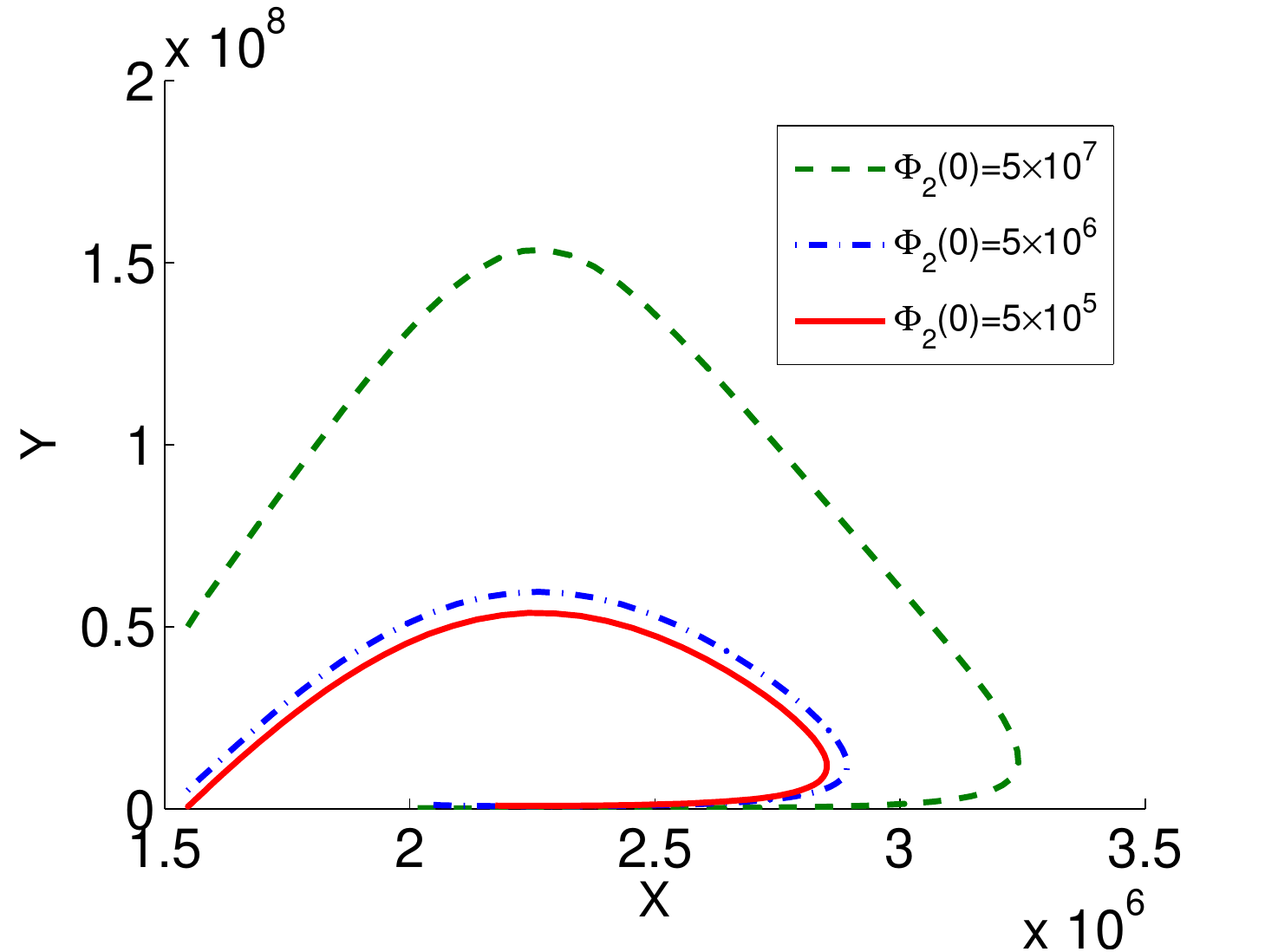}
\caption{
Trajectories of the non-scaled version of tumour-immune system~\eqref{system:onedelay} for initial condition as defined in \eqref{eq:initcond} with different values of $\Phi_2(0)$: $5\times10^5$ (dash green line), $5\times10^6$ (dash-dot blue line) and $5\times10^7$ (solid red line). All used parameter values are given in Table~\ref{table2} (column chimeric), $s\in[0,110]$. 
All solutions computed with RelTol=1e-7 and AbsTol=1e-11.}\label{fig:3}
\end{figure}

\begin{figure}[htb!]
\centering
\includegraphics[width=0.49\textwidth]{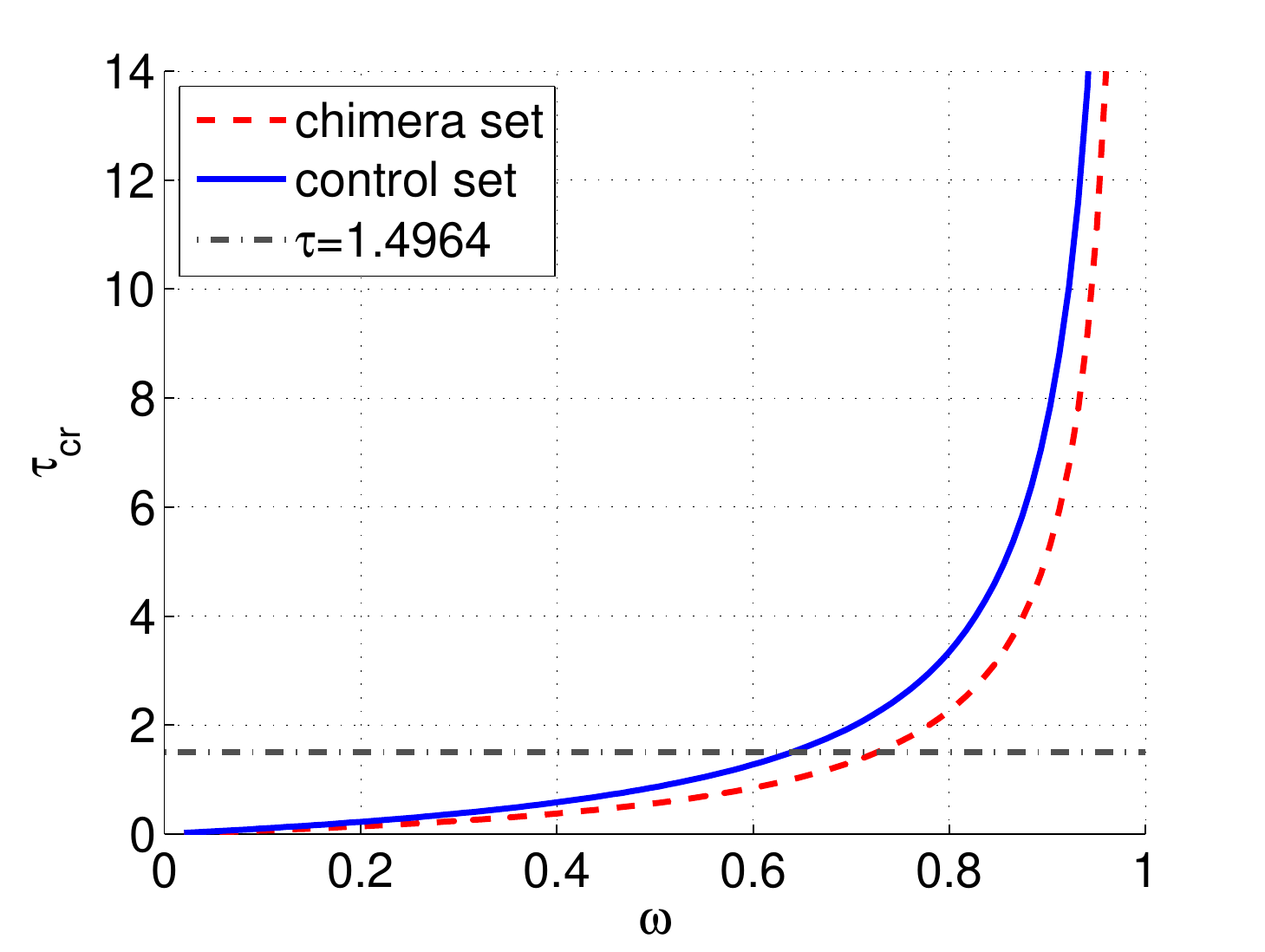}
\includegraphics[width=0.49\textwidth]{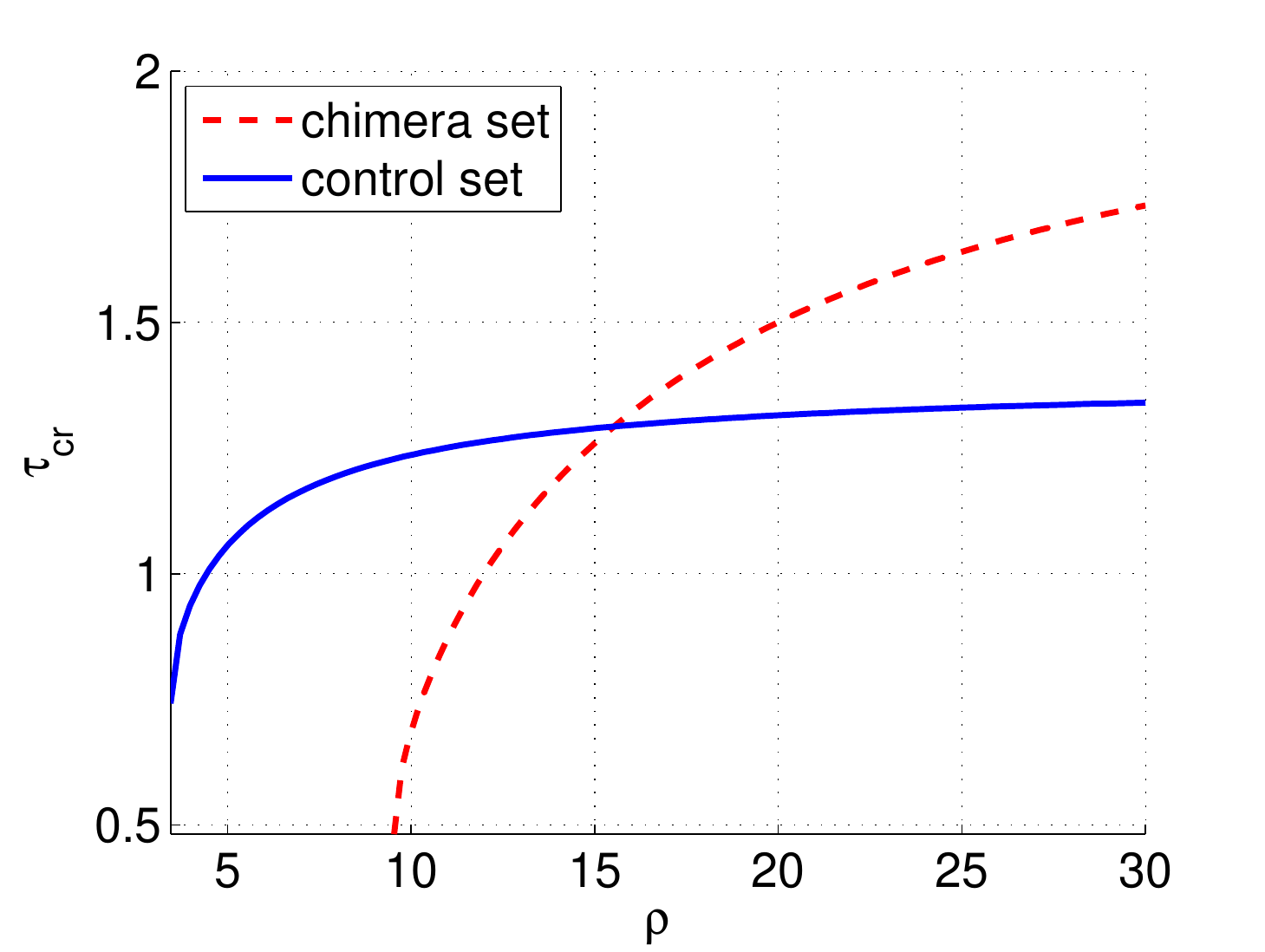}
\caption{(Left) Critical values of the bifurcation parameter as the functions of $\omega$ (left) and as the functions of $\rho$ (right) for both chimeric (red dash line) and control (blue solid line) sets of parameters. All used parameter values (except $\omega$, $\rho$ and $\tau$) are given in Table~\ref{table3}. The stability regions lie below the solid blue and the dashed red lines. The horizontal grey dash-dot line in the left panel indicates the value of delay estimated via the fitting procedure for the control experiment $\tau=1.4964$.}\label{fig:3a}
\end{figure}

For the sets of estimated parameters  that have been scaled and are given in Table~\ref{table3} the system~\eqref{system:onedelay} has one semi-trivial steady state which is equal to $(0.5802, 0)$  and $(0.6841,0)$ for the control and chimeric parameter sets, respectively.  
For both sets of parameters $\omega<1$, thus both semi-trivial steady states are unstable. Moreover, according to Theorem~\ref{trivialss} they stay unstable independently of the value of time delay.

Additionally, for both sets of parameters,
assumptions of Proposition~\ref{prop2} are fulfilled and thus,  for both cases  the
system~\eqref{system:onedelay} has two positive steady states: $B_{c}=(1, 0.0775)$ and 
$C_{c}=(1,3.7145)$ for control;  $B_{chim}=(1, 0.0422)$ and $C_{chim}=(1, 1.1745)$ for chimera.
Interesting is that the values of the positive steady states  $B_{chim}$ and $B_{c}$ are similar.  
For $\tau=0$ one can easily check that the steady states $B_{c}$ and $B_{chim}$ are 
locally asymptotically stable, while the steady states $C_{c}$ and $C_{chim}$ are unstable, see 
Table~\ref{table4}. Considering $\tau$ as a bifurcation parameter and using results of Theorem~\ref{pmn} 
one obtains that independently of the value of time delay the steady states $C_{c}$ and $C_{chim}$ remain 
unstable. The situation differs for steady states $B_{c}$ and $B_{chim}$, 
whose stability changes for 
sufficiently large values of parameter $\tau$. Moreover, Theorem~\ref{pmn} indicates that there can be only a single stability switch for each $B_c$ and $B_{chim}$ steady states for time delay treated as a bifurcation parameter.  For the considered sets of parameters the critical values of the bifurcation parameter are $1,1776$ 
and $1.2227$ for control and chimeric parameter sets, respectively.
Theorem~\ref{pmn} and the results of the numerical analysis indicate that
for sufficiently large time delay, the periodic behaviour of the solutions of system~\eqref{system:onedelay} should be observed if the initial condition is properly chosen.
Clearly, for the estimated chimeric set of parameters the delay value dose not exceed the critical value, thus for that set 
of parameters the  steady state $B_{chim}$ is locally asymptotically stable. For the control set of parameters, $B_c$ is unstable since for $\tau=1.1776$ the Hopf bifurcation takes place. 

The functional dependences of the critical values of the bifurcation parameter ($\tau_{cr}$) for the steady states 
$B_c$ and $B_{chim}$ on the model parameters $\omega$ and $\rho$ are presented in~Figure~\ref{fig:3a}, where all parameters 
are fixed (as in Table~\ref{table3}) except $\omega$, $\rho$ and $\tau$. Figure~\ref{fig:3a} left 
indicates that the critical values of time delays, for both  chimeric and control sets of parameters, are the 
increasing functions background immunity. Thus, the increase of the native immunity increases the stability region of 
the steady states $B_c$ and $B_{chim}$. Proposition~\ref{prop2} indicates that for $0<\omega<1$ (that is the case of the estimated 
values of parameters) and $\rho\rightarrow(\sqrt{\psi}+\sqrt{1-\omega})^2$ 
(that is $\rho\rightarrow 9.5355$ and $\rho\rightarrow 3.4419$, for chimeric and control sets of parameters, 
respectively) the positive steady states merge, thus the critical values $\tau_{cr}$ can only be computed for 
sufficiently large values of $\rho$. Additionally, there are no upper limits for $\rho$.
Similarly as in the previous case,  $\tau_{cr}$ for both sets of parameters are increasing functions of $\rho$. However, 
this time they are concave functions. Thus, 
increasing the stimulation strength enlarges the stability region of the 
steady states $B_c$ and $B_{chim}$, however in a different manner.   


In the general case a large number of model parameters makes the analytical
study of the nature of the Hopf bifurcations, i.e. of the stability of the arising periodic solutions, hard to perform. Thus, for the
estimated model parameters the BIFTOOL package  (a  Matlab package 
for numerical bifurcation analysis for DDEs developed by 
Engelborghs et al.) is used, for details and the documentation see~\cite{Engelborghs2001},
\cite{Engelborghs2002} or visit  http://twr.cs.kuleuven.be/research/software/delay/ddebiftool.shtml. 
In Figure~\ref{fig:4} the amplitudes of the periodic solutions as a function of the bifurcation parameter 
$\tau$ for the control and chimeric sets of parameters (see Table~\ref{table3}) for 
model~\eqref{system:onedelay} are presented. Clearly, from the figure one reads that periodic solutions 
are born for critical values of $\tau$ given in the Table~\ref{table4} and exist for $\tau$ greater than 
the critical values of the bifurcation parameter, indicating that the Hopf bifurcations are supercritical.
In both cases, the amplitudes are comparable, having the same order of magnitude. This result holds for 
the non-scaled model as well since the parameter $r$ is held constant across both cases, whilst the 
fitted value of $c$ has the same order of magnitude between cases.

\begin{figure}[htb!]
\centering
\includegraphics[width=0.40\textwidth]{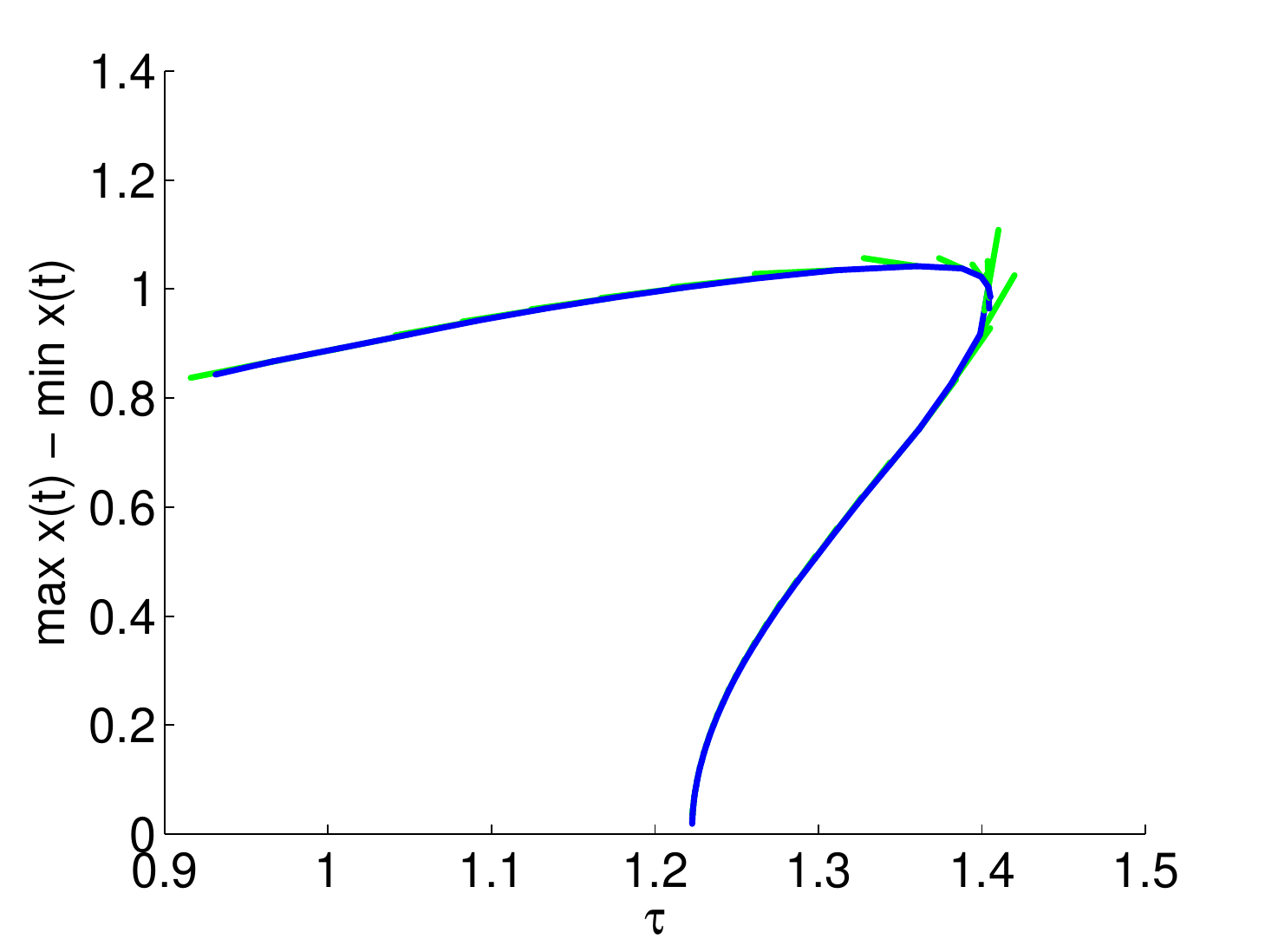}
\includegraphics[width=0.40\textwidth]{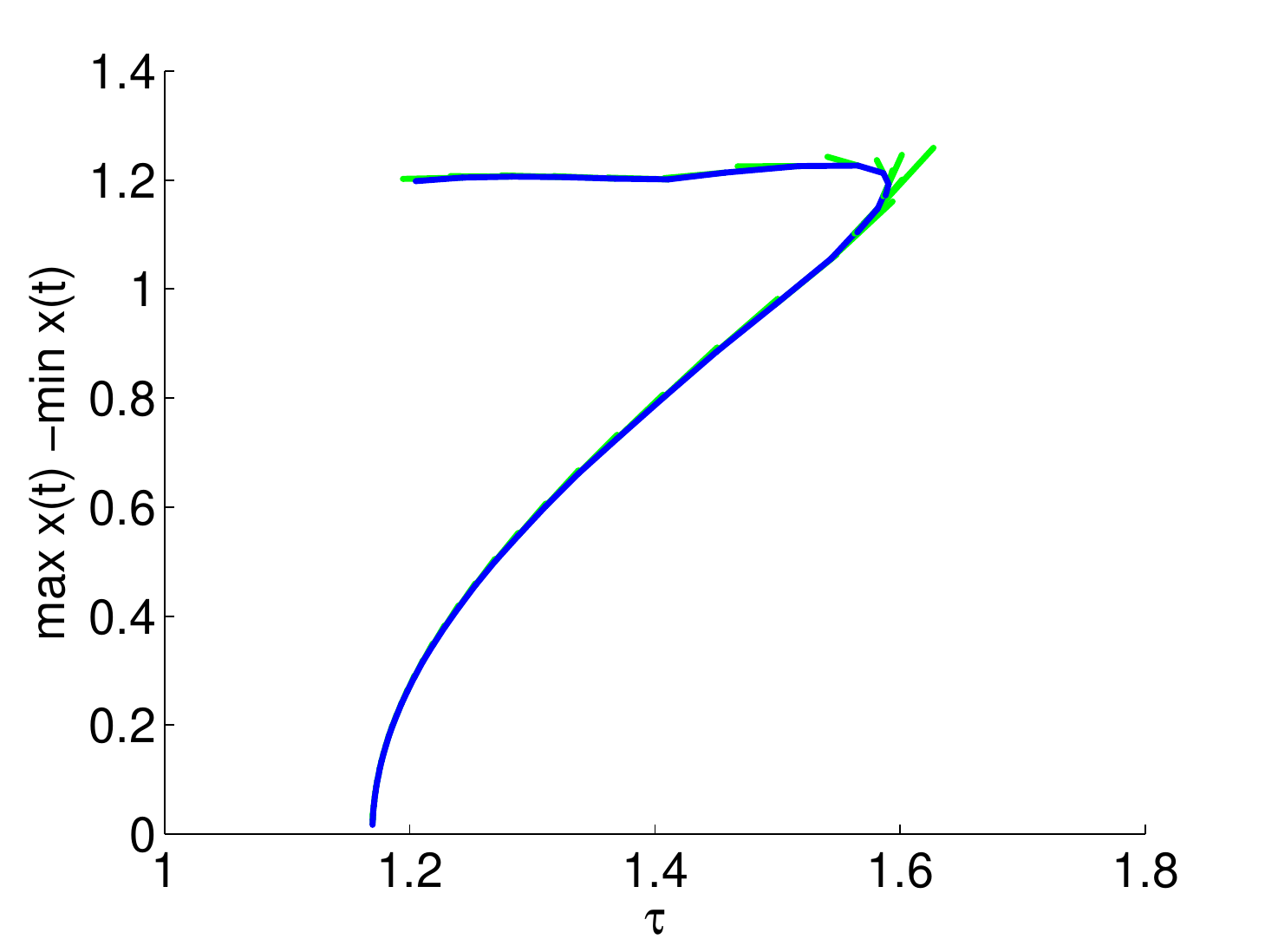}
\caption{Amplitudes of the periodic solutions arising due to the Hopf Bifurcation as a function of the 
bifurcation parameter $\tau$ for chimeric (left) and control (right) sets of parameters given in 
Table~\ref{table3} for model~\eqref{system:onedelay}. In green the predicted values are marked, while in 
blue the corrected 
ones.}\label{fig:4}
\end{figure}

\begin{figure}[htb!]
\centering
\includegraphics[width=0.40\textwidth]{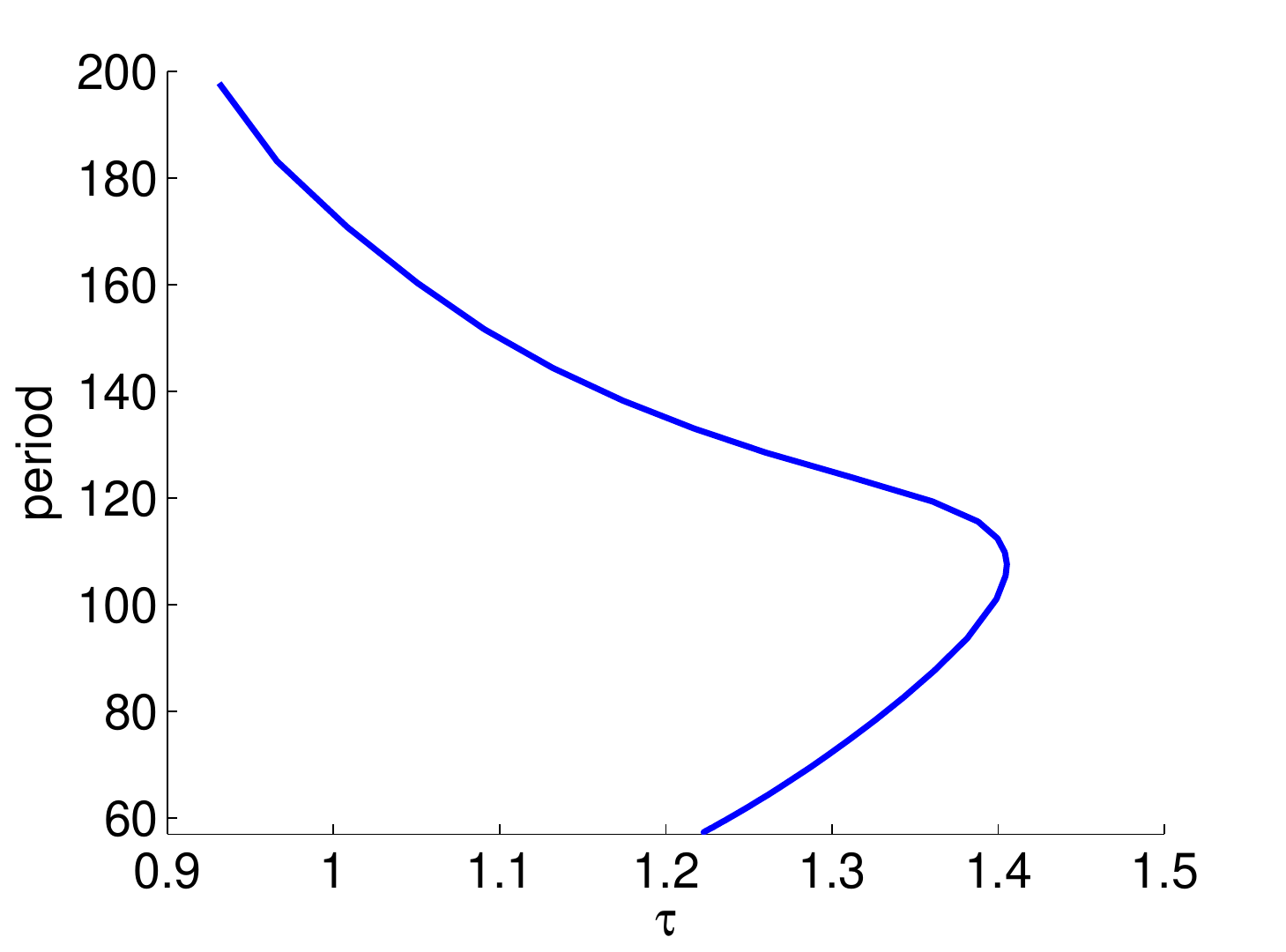}
\includegraphics[width=0.40\textwidth]{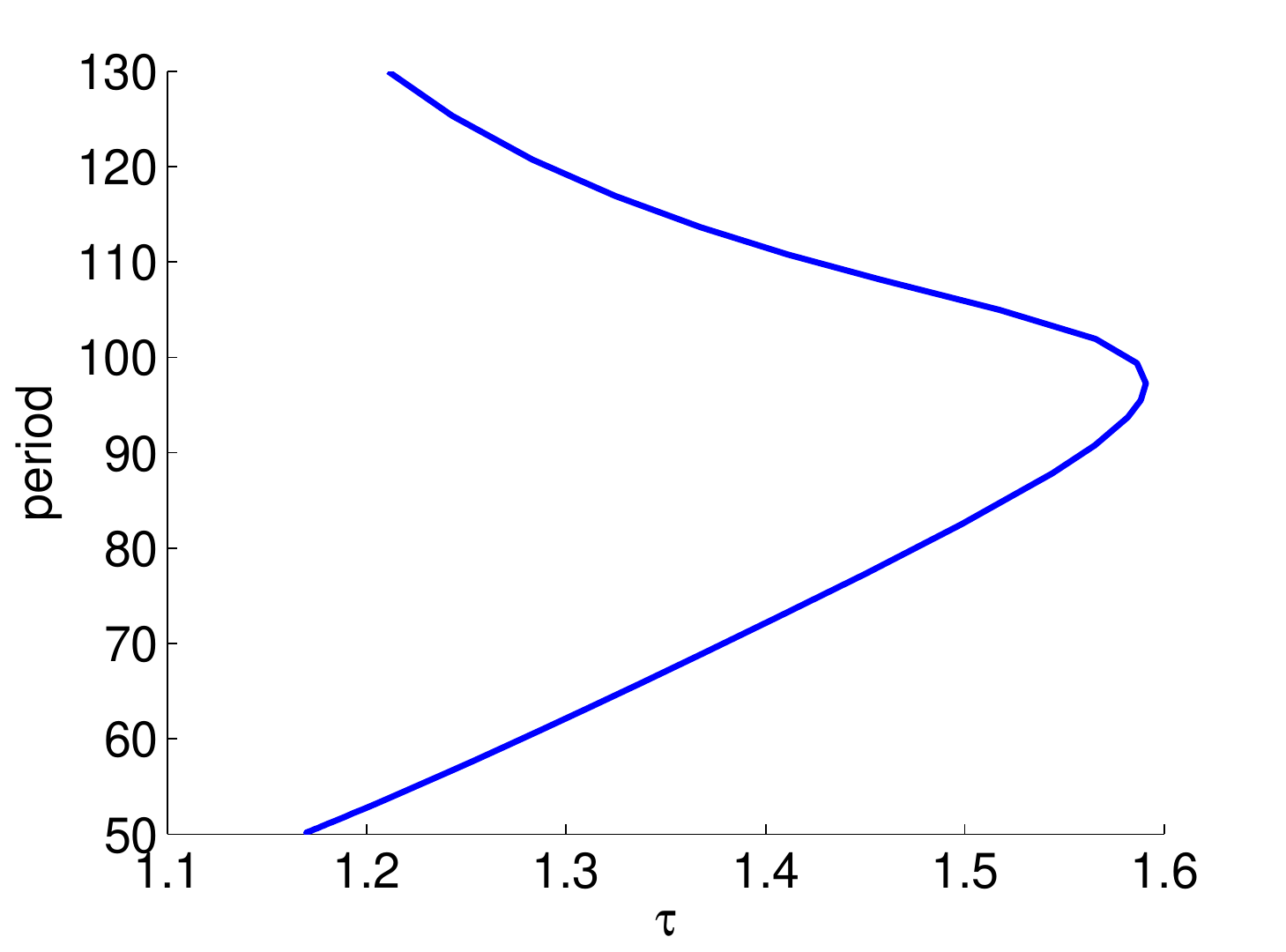}
\includegraphics[width=0.40\textwidth]{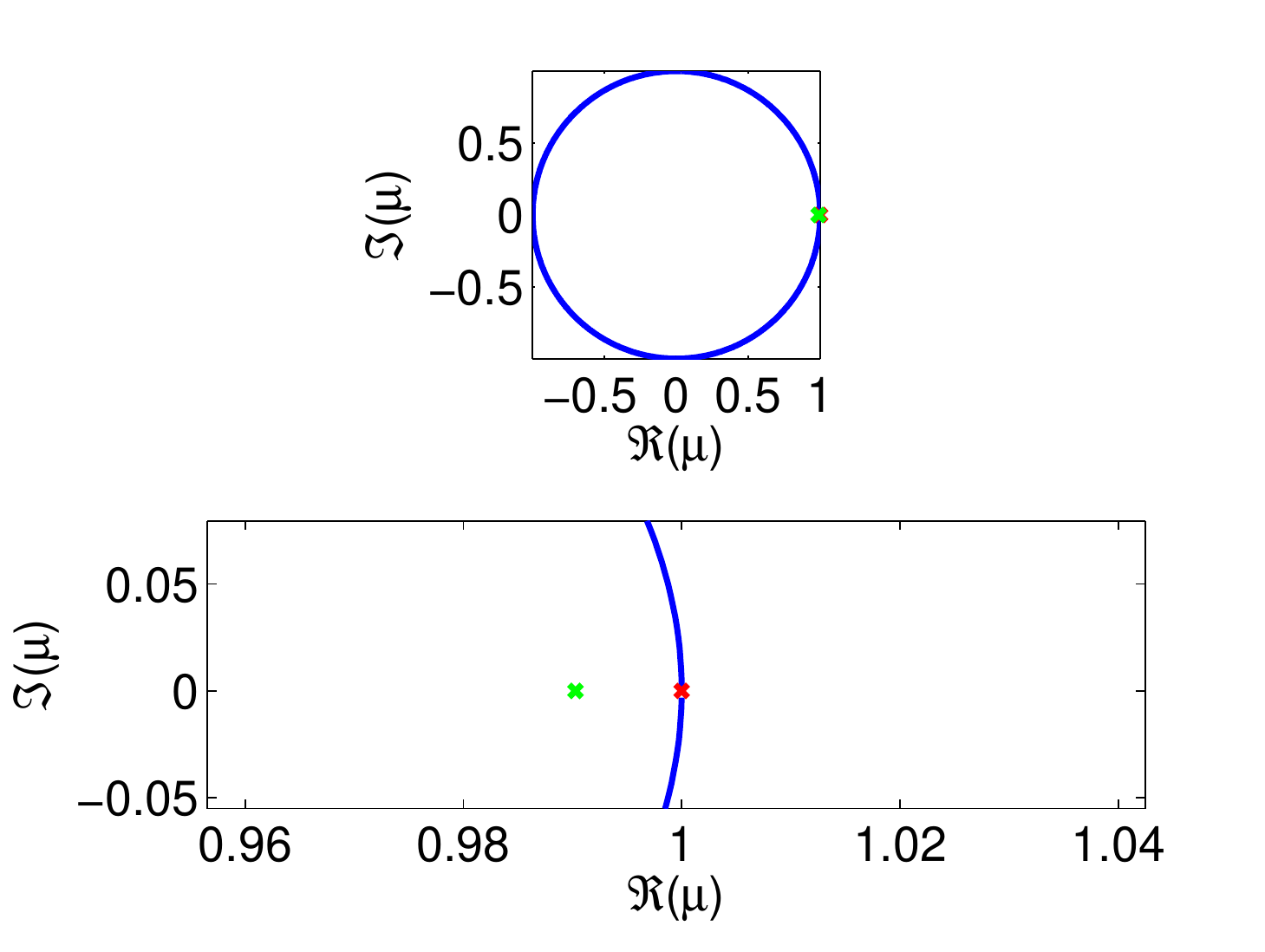}
\includegraphics[width=0.40\textwidth]{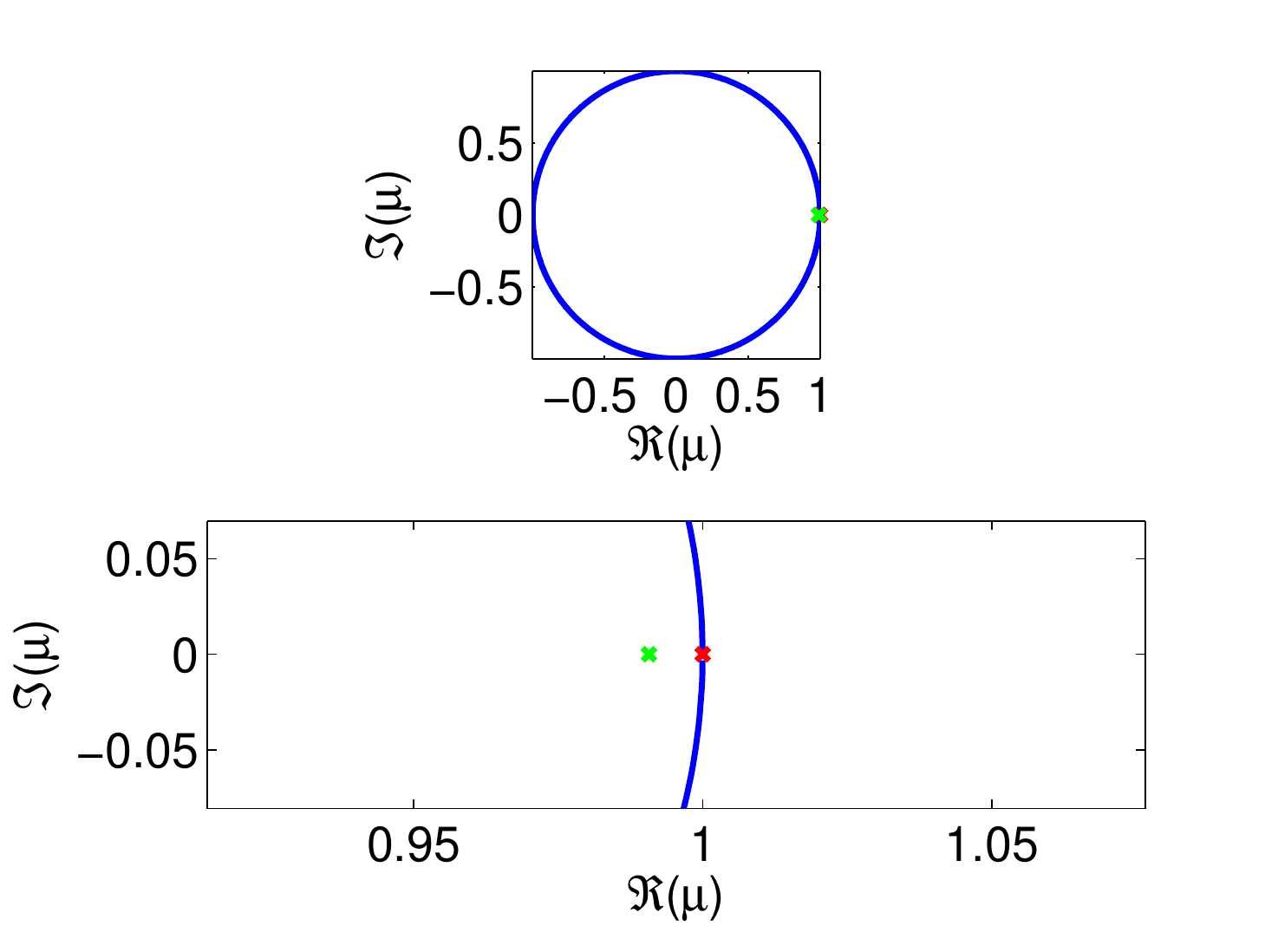}
\caption{Top: Periods of periodic solutions arising from the Hopf bifurcations as 
a function of  bifurcation parameter $\tau$ for chimeric set of parameters (left) and control set of parameters 
(right) for model~\eqref{system:onedelay}. 
Middle and bottom:  Floquet multipliers plotted on the complex plain for the periodic solutions near the 
bifurcation points for chimeric (left) and control (right) sets of parameters  for 
model~\eqref{system:onedelay}. In the bottom row the zoom of the results form the middle row is 
presented. For all simulations model parameters are given in Table~\ref{table3}.
}\label{fig:5}
\end{figure}

In Figure~\ref{fig:5} (top) the period of periodic solutions arising from the Hopf bifurcations as 
a function of  $\tau$ for control and chimeric sets of parameters are presented. 
The directly calculated (Theorem 3.3) values of the periods for critical values of bifurcation parameter 
are  $50.1363$ and $57.3715$ for control and chimeric sets of parameters, respectively, see 
Table~\ref{table4}. Calculating the same for the non-scaled model gives 
periods of approx. $107$ and $122$ days for the control and chimeric sets of parameters, respectively. These pattern agrees with 3-4 
months oscillations observed for the clinical data for certain kinds of human leukemias and also agrees 
with oscillatory growth of human and rat tumour cell lines reported in~\cite{Chignola2003,Chignola2000,Gliozzi2010} or in~\cite{Krikorian1980} for non-Hodgkin's lymphoma.

Interesting is that such recurrent patterns of tumour remission and re-growth have also been clinically observed for 
immunized BALB/c mice, where mice where initially immunized with BCL$_1$-Idiotype before the injection of $10^6$ BCL$_1$ 
tumour cells and without any further immunization,~\cite{Uhr1991}. For such experiment the level of Anti-Id in serum and spleen sizes 
where followed for a period of 115 days showing that in several mice the spontanous regression followed by 
splenomegaly was observed, compare Figure 5 in~\cite{Uhr1991}.      
Moreover, Figure~\ref{fig:5} (top) indicates  that the periods of the periodic solutions 
increase with the increase of the bifurcation parameter thus,
for larger time lags in the reaction of the immune system to the presence of antigens, larger periods of oscillation of the tumour mass and immunity are observed.

In Figure~\ref{fig:5} (middle and bottom) the Floquet multipliers plotted on the complex plain
for both sets of estimated parameters  are presented.
Clearly, the multiplier $(0,1)$ is the trivial one that always exists for the autonomous 
systems, while the ones inside the circle has modulus smaller than $1$. Since periodic solutions are stable if all, 
except (0,1), multipliers have modulus smaller than one, one concludes that for the parameter values 
given in Table~\ref{table3} the 
arising periodic solutions are stable. This implies that the supercritical Hopf bifurcations for both sets of 
parameters are observed, confirming our previous results presented in Figure~\ref{fig:4}. 
From the biological or medical point of view supercritical Hopf bifurcation is 
considered as a {\it safe bifurcation} (of course assuming that the amplitude of periodic solutions can not be 
too large), since although the 
steady state becomes unstable for larger values of the bifurcation parameter the new born periodic 
solutions are stable.
In such situations a local stable limit cycle around the dormant tumour steady state occurs 
and again some kind of boundedness of tumour growth is observed. Such oscillations around the steady 
state corresponds to situation where the tumour grows and shrinks with no application of additional 
treatment, see the experimental tumour volume data
from~\cite{Chignola2003,Chignola2000} or~\cite{Gliozzi2010}. 

%
\section{Discussion}\label{sec:diss}

In this paper  the interaction between  tumour and  immune system by a system of non-linear 
differential equations with discrete delay, representing time lag in the reaction of immune 
system to recognition of the antigens, is described.  The mathematical properties of proposed model such as: existence, uniqueness and non-negativity of solutions are studied. The presented study focusses on the  on the stability of the 
steady states depending on the value of the bifurcation parameter, which is the time delay.
The direct conditions for the stability of the steady states are derived in this paper.
In the general case, 
it is shown that depending on assumptions on the model parameters the stability of the system changes due to Hopf bifurcation for sufficiently large delay.
Moreover, only one such stability switch is possible.  Additionally, it is shown that the unstable (for $\tau=0$) steady states remain unstable independently of the values of delay. 

Next, the proposed model is calibrated with two sets of the experimental data for BCL$_1$ tumour cells injected to the mice spleens,~\cite{Uhr1991}, with MSEs at the 
level of  3.57\% and 0.14\%. For the estimated values of parameters, the proposed model is further validated with other 
experiments for different initial numbers of injected cells.
For both sets of estimated values of the parameters the considered model has two positive steady states and one 
unstable semi-trivial steady state. For the chimeric set of parameters there are stable and unstable positive steady 
states, which indicate a local stabilization (dormancy) of tumour growth at finite size if 
the initial size of tumour is sufficiently small (as in case of considered experimental setup) or infinite tumour growth 
otherwise. 
For the estimated control set of parameters both positive steady states are unstable, however the periodic orbits 
exist. That means that depending on the initial conditions the unrestricted tumour growth (as for the control set of parameters) or periodic oscillations of tumour size are observed.   
  
The second coordinates of the positive stable steady state  for the chimeric set of the parameters and 
positive steady 
state (that is stable for sufficiently small delay) for the control set of parameters are similar.
Hence, by decreasing the time response of the immune system to the presence of the tumour cells one might stabilize the 
uncontrolled tumour growth in case of the control experiment.

Additionally, for both estimated sets of parameters (except time delay treated as a bifurcation parameter) 
the presented analytical study is extended by investigating the 
stability of the new born periodic orbits and thus the type of existing Hopf bifurcations.
It should be pointed out that the oscillatory behaviour of the multicellular tumour volume was experimentally 
observed, even without administration of any treatment, as reported for the {\it in vitro} experiments for human and rat cell lines
reported in~\cite{Chignola2000} and~\cite{Chignola2003}.  
Such oscillations are explained e.g. as the effect 
of non-linear feedback between proliferation at the multicellular tumours surface and invasion of the 
surrounding by tumour cells leaving the spheroid surface however other scenarios are also discussed. 
The oscillating growth phenomenon is also observed for the {\it 
in vivo} studies and take place for e.g. the mice initially immunized with BCL$_1$-Idiotype before the injection 
of $10^6$ BCL$_1$ tumour cells or in non-Hodgkin's lymphoma~\cite{Krikorian1980}.

Although the model is rather simple,  its dynamics might be complex and reflect a number of bio-medical phenomena such as: therapy failure considered as a modification of model parameters, i.e. an unrestricted growth of the tumour; stabilisation (at least 
temporary) of tumour growth at a finite size -- dormancy; and desirable regression of the 
tumour. 
If the initial size of the tumour (or rather amount of antigens) is strictly greater than zero, then  according to the
values of the model parameters seven different scenarios are possible: (i)
asymptotic unrestricted growth of the tumour independently of the initial tumour size; (ii)
stabilisation of the tumour growth  independently of initial tumour size (existence of globally stable 
equilibrium); (iii)
global eradication of the tumour by the immune system independently of the initial tumour size;
(iv)
local stabilisation of tumour growth at a certain finite size (if the initial tumour size is sufficiently small one observes 
stabilisation of the tumour growth at a finite size, otherwise unrestricted growth); (v)
local eradication of the tumour by the immune system (for sufficiently small initial size tumour will be  
eradicated by the immune system, otherwise it will grow unrestrictedly); (vi)
bi-stability with eradication for sufficiently small  initial tumour, stabilisation on a finite size of tumour of medium 
initial size and unrestricted growth for large initial tumour when two locally stable steady states
co-exist; and finally (vii) recurrent patterns of tumour remission and re-growth. The first six of these features can be  
observed for the model without time delay, while the last one it is not possible since the ODE model has no periodic 
solutions. Hence, the DDE model can reflect the oscillatory 
phenomenon observed for a range of tumour types and in particular in lymphoma.

The main goal of the any-anticancer therapies is to treat disease or at least control its progression. Form the 
mathematical point of view this means that one wishes to ensure that the trajectories of the system are in the region of 
attraction of the proper steady states.  To do that at the level of the considered model one can try to increase or 
decrease some of the parameters. 
In the context of immunotherapy several attempts to increase the immune system efficiency or  the creation of 
anti-tumour immunity were already proposed in 
many different ways \cite{Schreiber1999, Levy1999, Lollini2002}.  
One can consider the increase of the native immunity e.g. by proper incubation of a patient's 
natural killer lymphocytes to enhance their activity and later by their infusion of them into the patient body. In the 
mathematical language that would mean the increase of the $\omega$ parameter which is crucial for reaching 
the global or local eradication of the tumour or the bi-stability phenomenon (depending on the values of the 
other system parameters).   Recall that $\omega=wc/(ur)$. Clearly, since the natural lymphocyte death 
rate ($u$) and the tumour growth rate ($r$) are rather out of reach of immunotherapy one should thus focus on
increasing $w$ and $c$ parameters, that is  stimulation of the production strength and the effectiveness of effector cells. The increase 
of the effective cells can be reached by the infusions of the cells grown {\it ex vivo}, while the increase of the 
killing effectiveness of the cells can be obtained as the result of the infusion of properly chosen 
cytokines,~\cite{Schreiber1999}.  On the other hand, parameter $r$ can be for example decreased by different types of 
chemotherapies. 
Another parameter that has an essential role in determining the range of the controlled tumour growth is the 
stimulation strength $\rho=a/u$. If the tumour antigens are not presented to the lymphocytes, one 
can  increase $a$ by method of artificial loading of the of antigen-presenting calls with tumour antigen. 
The effectiveness of such an approach will be dictated by the values of the background immunity and strength of the 
anti-immune activity.
In this paper,  it is shown that for the model with discrete delay and estimated parameters an increase 
of the background immunity ($\omega$), as well as the stimulation strength ($\rho$) of the tumour 
cells on the immune system, causes  the increase of the critical delay value for which the Hopf 
bifurcation is observed, enlarging the stability region for the positive dormant steady state, see 
Figure~\ref{fig:3a}. Hence, by increasing the background immunity (range is reachable) one might stabilize the  
unrestricted 
 tumour growth observed for the control experiment. 

On the other hand, if one wishes to model the effectiveness of a given therapy or try to 
increase a given therapeutic effect
one should rather directly include these elements in the system, e.g. by adding new terms and considering the optimization problem.
That of course would force the additional calibration of the 
extended model with additional sets of experimental data required.
Researchers who are interested in collaborating with experimentalists or clinicians who have access to patients or {\it in vivo} animal trial data should feel encouraged to follow this idea.
Recently the very promising concept of an anti-tumour 
vaccines has arisen,~\cite{Schreiber1999,Levy1999}.
However, due to the discrete nature of the vaccination program, the presented model would need to be adapted to include an impulse-response framework.
In that context,  the presented calibrated 
model can be a base for the future study of the effectiveness of a different types of therapies, including the 
chemotherapy or combined chemo-immunotherapy.
One should feel free to modify the proposed model by fitting it to appearing needs. For example one can consider 
instead of the exponential tumour growth rather the Gompertz one, which takes into account the finite 
environmental capacity for tumour cells. That might cause an essential qualitative difference between models. 
However, it is still an open question.  On the other hand, to fit even better proposed model with the 
experimental data one could instead of the discrete delay consider the distributed one that often better reflects 
the real measurements.


\section*{Acknowledgements}
The work on this paper was supported
by the Polish Ministry of Science and Higher Education, within the Iuventus Plus Grant: ''Mathematical
modelling of neoplastic processes" grant No. IP2011 041971.
\nolinenumbers


\bibliographystyle{elsarticle-num}

\end{document}